\documentclass[sigconf]{acmart}
\usepackage[utf8]{inputenc}

\newif\iffull
\fulltrue

\newif\ifarxivsource
\arxivsourcetrue

\usepackage{amsmath,amsthm}
\usepackage[at]{easylist}
\usepackage{xspace}

\usepackage{booktabs} %
\usepackage{caption}
\usepackage{subcaption}
\usepackage{enumerate}
\usepackage[normalem]{ulem} %

\usepackage{booktabs}

\usepackage{paralist}

\makeatletter
\newcommand{\removelatexerror}{\let\@latex@error\@gobble}
\makeatother

\definecolor{darkgreen}{rgb}{0,0.6,0}
\newcommand{\ignore}[1]{}

\newcommand{\smath}[1]{\ensuremath{#1}\xspace}

\newcommand{\TMstyle}[1]{\smath{\mathcal{#1}}}
\newcommand{\methodstyle}[1]{\smath{\textsc{#1}}}
\newcommand{\privatedatastyle}[1]{\smath{\mathit{#1}}}
\newcommand{\publicdatastyle}[1]{\smath{\textsc{#1}}}
\newcommand{\cryptoobjectstyle}[1]{\smath{#1}}
\newcommand{\algostyle}[1]{\smath{\textsf{#1}}}
\newcommand{\boolstyle}[1]{\smath{\textsc{#1}}}

\theoremstyle{definition}
\newtheorem{theorem}{Theorem}[section]
\newtheorem{definition}[theorem]{Definition}
\newtheorem{lemma}[theorem]{Lemma}
\newtheorem{claim}[theorem]{Claim}
\newtheorem{remark}[theorem]{Remark}

\ifarxivsource
    \renewcommand{\G}{\TMstyle{G}}
    \renewcommand{\C}{\TMstyle{C}}
\else
    \newcommand{\G}{\TMstyle{G}}
    \newcommand{\C}{\TMstyle{C}}
\fi
\newcommand{\R}{\TMstyle{R}}
\newcommand{\V}{\TMstyle{V}}
\newcommand{\xA}{\TMstyle{A^\ast}}
\newcommand{\A}{\TMstyle{A}}
\newcommand{\rA}{\A}
\newcommand{\E}{\TMstyle{E}}
\newcommand{\N}{\TMstyle{N}}
\newcommand{\T}{\TMstyle{T}}
\renewcommand{\P}{\TMstyle{P}}
\renewcommand{\prec}{\preccurlyeq}

\newcommand{\trans}{\cryptoobjectstyle{\tau}} %
\newcommand{\exec}[2]{\smath{\langle #1, #2 \rangle}}

\newcommand{\Accept}{\boolstyle{Accept}}
\newcommand{\Reject}{\boolstyle{Reject}}
\newcommand{\True}{\boolstyle{true}}
\newcommand{\False}{\boolstyle{false}}

\newcommand{\prompt}{\methodstyle{prompt}}
\newcommand{\Read}{\methodstyle{read}}
\newcommand{\Write}{\methodstyle{write}}

\newcommand{\deny}{\methodstyle{deny}}

\newcommand{\loc}[1]{{\publicdatastyle{{#1}Loc}}}

\newcommand{\exDecom}{\textsf{decom}}
\newcommand{\enterPwd}{{\textsf{pwd}}}
\newcommand{\mfa}{\textsf{2fa}}
\newcommand{\hash}{\textsf{hash}}
\newcommand{\rand}{\textsf{rand}}
\newcommand{\lang}{{\textsf{lang}}}
\newcommand{\hybrid}{{\textsf{read}}}
\newcommand{\readwrite}{{\textsf{readWrite}}}
\newcommand{\secret}{{\textsf{secret}}}
\newcommand{\known}{{\textsf{known}}}
\newcommand{\denysub}{{\textsf{deny}}}
\newcommand{\binding}{{\textsf{bind}}}

\newcommand{\msg}{\privatedatastyle{{m}}}

\newcommand{\pwd}{\privatedatastyle{{pwd}}}
\newcommand{\decrypted}{\privatedatastyle{decrypted}}\newcommand{\duress}{\privatedatastyle{duress}}

\newcommand{\Com}{\algostyle{Com}}
\renewcommand{\Check}{\algostyle{Check}}
\newcommand{\OTP}{\algostyle{OTP}}
\newcommand{\Enc}{\algostyle{Enc}}

\newcommand{\thecode}{\privatedatastyle{code}}
\newcommand{\gotpwd}{\privatedatastyle{gotPwd}}

\newcommand{\getcode}{\methodstyle{getCode}}
\newcommand{\setcode}{\methodstyle{setCode}}
\newcommand{\findsecondary}{\methodstyle{findSecond}}

\newcommand{\promptpw}{\methodstyle{promptPwd}}
\newcommand{\promptcode}{\methodstyle{promptCode}}

\newcommand{\findfile}{\methodstyle{findFile}}

\usepackage[boxed,titlenotnumbered,noend,nofillcomment]{algorithm2e}
\DontPrintSemicolon
\SetTitleSty{textsc}{11pt}
\SetKwRepeat{Rcv}{receive}{from} %
\SetKwRepeat{Send}{send}{to} %
\SetKw{Call}{call} %
\SetKw{Set}{set} %
\SetKw{Retrieve}{retrieve} %
\SetKwRepeat{Let}{let}{s.t.} %
\SetKwInput{Method}{Method} %
\SetKwFor{MethodDef}{Method}{}{}
\SetKwInput{Variable}{Variables}
\SetKwInput{ConnectsTo}{Channel to}
\SetKwInput{Constant}{Constant}
\SetKwInput{Oracle}{Oracle}
\SetKwInput{OracleR}{Oracle (read only)}
\SetKwInput{OracleRW}{Oracle (read/write)}
\SetKwInput{Assert}{assert}
\SetKwFor{AssertAfter}{assert}{after}{} %
\SetKwFor{AssertWhere}{assert}{where}{} %
\newcommand{\algoTitle}[2]{\underline{\textsc{#1} #2}:\;\vspace{.3em}}

\newcommand{\triplealgorithm}[4]{
    \begin{figure*}[ht]
        \begin{subfigure}[ht]{.325\textwidth}
            \removelatexerror %
            #1
        \end{subfigure}%
        \hfill
        \begin{subfigure}[ht]{.325\textwidth}      
            \removelatexerror
            #2
        \end{subfigure}%
        \hfill
        \begin{subfigure}[ht]{.325\textwidth}
            \removelatexerror
            #3
        \end{subfigure}
    \caption{#4}
    \end{figure*}
}

\newcommand{\exampleITM}{%
    \removelatexerror %
    \begin{algorithm}[H]
    \algoTitle{Example ITM}{$M$}
    \Variable{$x_1$, $x_2\neq 0$, $x_3 \gets 5$\;}
    \MethodDef{$\methodstyle{Set}(i,x')$}{\Set{$x_i \gets x'$}}
    \MethodDef{$\methodstyle{Send()}$}{
    $M'$.$\methodstyle{Receive}(x_1, x_2, x_3)$
    }
\end{algorithm}}

\newcommand{\decomverifier}{%
    \removelatexerror
    \begin{algorithm}[H]
        \algoTitle{Verifier}{$\V_{\exDecom}^{\N}$}
    	\lRcv{$\A$}{$(x', d')$}
    	$c \leftarrow \N[\E.\publicdatastyle{commLoc}].\methodstyle{read}()$\;
    	\If{$(\Check(c, d', x') == 1)$}{\Return \Accept}
    	\lElse{\Return \Reject}
    \end{algorithm}
}

\newcommand{\decomevidence}{%
\removelatexerror
\begin{algorithm}[H]
    \algoTitle{Evidence}{$\E_{\exDecom}$}
    \KwData{$\publicdatastyle{commLoc}$\;}
    \OracleR{$\N[\publicdatastyle{commLoc}]$\;}
    \Method{$\R.\methodstyle{secret}()$, $\R.\methodstyle{decom}()$}
    \AssertAfter{$\Check(c, d, x) = 1$}{$c \leftarrow \N[\publicdatastyle{commLoc}]$\; $d \leftarrow \R.\methodstyle{decom}()$\; $x \leftarrow \R.\methodstyle{secret}()$\;
    $\exists \Com$ such that $(\Com,\Check)$ is a commitment scheme \quad\quad($\star$)}
\end{algorithm}}

\newcommand{\decomtarget}{%
\removelatexerror
\begin{algorithm}[H]
    \algoTitle{Target action}{$\T_{\exDecom}^{\R,\N}$}
    \Return $\R.\methodstyle{secret}()$
\end{algorithm}}

\newcommand{\decompost}{%
\removelatexerror
\begin{algorithm}[H]
    \algoTitle{Post-processor}{$\P_{\exDecom}^{\N'}(\trans)$}
    \text{Parse the first round of }$\trans$ \text{as} {$(x, d)$}\;
    \Return $x$
\end{algorithm}}

\newcommand{\decomexemplar}{%
\removelatexerror
\begin{algorithm}[H]
    \algoTitle{Exemplar action}{$\xA_{\exDecom}^{\R,\N}$}
    $x \leftarrow \R.\methodstyle{secret}()$\;
    $d \leftarrow \R.\methodstyle{decom}()$\;
    \lSend{$\V$}{$(x,d)$}
\end{algorithm}}

\newcommand{\eptdD}{%
\removelatexerror
\begin{algorithm}[H]
    \algoTitle{Device}{$D_\enterPwd$}
    \Variable{$\pwd$, $\msg$, $\decrypted \gets \False$\;}
    \MethodDef{$\prompt(x)$}{
        \lIf{$(x==\pwd)$}{
            \Set{$\decrypted \gets \True$}
        }
    }
    \MethodDef{$\Read()$}{
        \lIf{$(\decrypted == \True)$}{\Return{$\msg$}}
        \lElse{\Return{$\bot$}}
    }
\end{algorithm}}

\newcommand{\eptdDdeny}{%
\removelatexerror
\begin{algorithm}[H]
    \algoTitle{Device}{$D_\denysub$ implements $D_\enterPwd$}
    \tcp{Has the following variables and methods in addition to those provided by $D_\enterPwd$ in Fig.~\ref{algs:enter-a-password}}
    \Variable{$\duress$\;}
    \MethodDef{$\deny(x,y)$}{
        \lIf{$(x==\duress)$}{
            \Set{$\msg \gets y$\;}
            \Set{$\decrypted \gets \True$}
        }
    }
\end{algorithm}}

\newcommand{\eptdE}{%
\removelatexerror
\begin{algorithm}[H]
    \algoTitle{Evidence}{$\E_\enterPwd$}
    \KwData{$\loc{D}$;}
    \Oracle{$\N[\loc{D}]$\;}
    \Method{$\R.\pwd()$\;}
    \Assert{$D \prec \N[\loc{D}]$; $D.m\neq \bot$;
    $\R.\pwd() == D.\pwd$ \quad\quad{($\star$)} \;}
\end{algorithm}}

\newcommand{\eptdV}{%
\removelatexerror
\begin{algorithm}[H]
    \algoTitle{Verifier}{$\V_\enterPwd^{\N}$}
    \Set{$\msg \gets \N[\E.\loc{D}].\Read()$}\;
    \lIf{$(m\neq\bot)$}{\Return{\Accept}}
    \lElse{\Return{\Reject}}
\end{algorithm}}

\newcommand{\eptdT}{%
\removelatexerror
\begin{algorithm}[H]
    \algoTitle{Target}{$\T_\enterPwd^{\N,\R}$}
    \Set{$x \gets \R.\pwd()$}\;
    $\N[\E.\loc{D}].\prompt(x)$\;
    \Set{$\msg' \gets \N[\E.\loc{D}].\Read()$}\;
    \Return{$\msg'$}
\end{algorithm}}

\newcommand{\eptdX}{%
\removelatexerror
\begin{algorithm}[H]
    \algoTitle{Exemplar}{$\xA_\enterPwd^{\N,\R}$}
    $\N[\E.\loc{D}].\prompt(\R.\pwd())$
\end{algorithm}}

\newcommand{\eptdP}{%
\removelatexerror
\begin{algorithm}[H]
    \algoTitle{Post-processor}{$\P_\enterPwd^{\N'}(\trans)$}
    \Return{$\N[\E.\loc{D}].\Read()$}
\end{algorithm}}

\newcommand{\mfaD}{%
\removelatexerror
\begin{algorithm}[H]
    \algoTitle{Primary Device}{$D$}
    \Variable{$\pwd$, $\msg$, $\thecode$, $\decrypted \gets \False$, $\gotpwd \gets \False$\; }
    
    \MethodDef{$\promptpw(x)$}{
        \If{$(x==\pwd)$}{
            \Set{$\thecode \gets \$$}\;
            $S.\setcode(\thecode)$\;
            \Set{$\gotpwd \gets \True$}
        }
    }
    \MethodDef{$\promptcode(c)$}{
        \lIf{$(\gotpwd == \True) \land (c==\thecode)$}{
           \!\!\Set{$\decrypted \gets \True$}
        }
    }
    \MethodDef{$\Read()$}{
        \lIf{$(\decrypted == \True)$}{\Return{$\msg$}}
        \lElse{\Return{$\bot$}}
    }
\end{algorithm}}

\newcommand{\mfaE}{%
\removelatexerror
\begin{algorithm}[H]
    \algoTitle{Evidence}{$\E_\mfa$}
    \KwData{$\loc{device}$\;}
    \Method{$\R.\methodstyle{pwd}()$, $\R.\findsecondary()$\;}
    \Oracle{$\N[\loc{device}]$, $\N[\R.\findsecondary()]$ \;}
    \Assert{$D\prec\N[\loc{device}]$; $S\prec\N[\R.\findsecondary()]$; $\R.\methodstyle{pwd}() == D.\pwd$; $D.\msg \ne \bot$\;}
    \AssertAfter{$D.code == c$}{
        $D.\promptpw(D.\pwd)$\;
        $c\gets S.\getcode()$\;
    }
\end{algorithm}}

\newcommand{\mfaS}{%
\removelatexerror
\begin{algorithm}[H]
    \algoTitle{Secondary Device}{$S$}
    \Variable{$c$}
    \MethodDef{$\setcode(\thecode)$}{
    \Set{$c \gets \thecode$}
    }
    \MethodDef{$\getcode()$}{
    \Return{$c$}
    }
\end{algorithm}}

\newcommand{\mfaV}{%
\removelatexerror
\begin{algorithm}[H]
    \algoTitle{Verifier}{$\V_\mfa^{\N}$}
    $\msg \gets \N[\E.\loc{device}].\Read()$\;
    \lIf{$m\neq\bot$}{\Return{\Accept}}
    \lElse{\Return{\Reject}}
\end{algorithm}}

\newcommand{\mfaT}{%
\removelatexerror
\begin{algorithm}[H]
    \algoTitle{Target}{$\T_\mfa^{\N,\R}$}
    $x \leftarrow \R.\methodstyle{pwd}()$
    $\N[\E.\loc{device}].\promptpw(x)$\hspace*{-3em}\;
    $c \gets \N[\R.\findsecondary()].\getcode()$\;
    $\N[\E.\loc{device}].\promptcode(c)$\!\;
    $\msg \gets \N[\E.\loc{device}].\Read()$\;
    \Return{$\msg$}
\end{algorithm}}

\newcommand{\mfaX}{%
\removelatexerror
\begin{algorithm}[H]
    \algoTitle{Exemplar}{$\xA_\mfa^{\N,\R}$}
    $\N[\E.\loc{device}].\promptpw(\R.\methodstyle{pwd}())$\hspace*{-3em}\;
    $c \gets \N[\R.\findsecondary()].\getcode()$\;
    $\N[\E.\loc{device}].\promptcode(c)$
    \lSend{$\V$}{$(x,d)$}
\end{algorithm}}

\newcommand{\mfaP}{%
\removelatexerror
\begin{algorithm}[H]
    \algoTitle{Post-processor}{$\P_\mfa^{\N'}(\trans)$}
    \Return{$\N[\E.\loc{primary}].\Read()$}
\end{algorithm}}

\newcommand{\langE}{%
\removelatexerror
\begin{algorithm}[H]
    \algoTitle{Evidence}{$\E_\lang$}
    \Variable{$\R.z$\;}
    \Method{$\R.x()$\;}
    \AssertWhere{$\R.x() \in L_{\R.z}$}{
    each $z' \in \{0,1\}^*$ is associated with a language $L_{z'} \subseteq \{0,1\}^*$}
\end{algorithm}}

\newcommand{\langT}{%
\removelatexerror
\begin{algorithm}[H]
    \algoTitle{Target}{$\T_\lang^{\R}$\hspace*{-3em}}
    \Return{$\R.x()$}
\end{algorithm}}

\newcommand{\hashE}{%
\removelatexerror
\begin{algorithm}[H]
    \algoTitle{Evidence}{$\E_\hash$}
    \KwData{$y$\;}
    \Method{$\R.\findfile()$\;}
    \OracleR{$\N[\R.\findfile()]$}
    \AssertAfter{$y == h(x)$}{
        $x \gets \N[\R.\findfile()].\Read()$\;
    }
\end{algorithm}}

\newcommand{\hashV}{%
\removelatexerror
\begin{algorithm}[H]
    \algoTitle{Verifier}{$\V_\hash^{\N}$}
    \lRcv{$\A$}{$x'$}
    \lIf{$(h(x') == \E.y)$}{\Return{\Accept}}
    \lElse{\Return{\Reject}}
\end{algorithm}}

\newcommand{\hashP}{%
\removelatexerror
\begin{algorithm}[H]
    \algoTitle{Post-processor}{$\P_\hash^{\N'}(\trans)$}
    Parse the first round of $\trans$ as $x'$\;
    \Return{$x'$}
\end{algorithm}}

\newcommand{\hashX}{%
\removelatexerror
\begin{algorithm}[H]
    \algoTitle{Exemplar}{$\xA_\hash^{\N,\R}$}
    $x \gets \N[\R.\findfile()].\Read()$\;
    \lSend{\V}{$x$}
\end{algorithm}}

\newcommand{\hashT}{%
\removelatexerror
\begin{algorithm}[H]
    \algoTitle{Target}{$\T_\hash^{\N,\R}$}
    \Return{$\N[\R.\findfile()].\Read()$}
\end{algorithm}}

\newcommand{\hybridE}{%
\removelatexerror
\begin{algorithm}[H]
    \algoTitle{Evidence}{$\E_\hybrid$}
    \KwData{$\loc{D}$\;}
    \Oracle{$\N[\loc{D}]$\;}
    \Assert{$D_\hybrid \prec \N[\loc{D}]$}
\end{algorithm}}

\newcommand{\hybridD}{%
\removelatexerror
\begin{algorithm}[H]
    \algoTitle{Device}{$D_\hybrid$}
    \Variable{$\msg$\;}
    \MethodDef{$\Read()$}{
        \Return{$\msg$}
    }
\end{algorithm}}

\newcommand{\hybridDprime}{%
\removelatexerror
\begin{algorithm}[H]
    \algoTitle{Device}{$D_\readwrite$}
    \Variable{$\msg$\;}
    \MethodDef{$\Read()$}{
        \Return{$\msg$}
    }
    \MethodDef{$\Write(x)$}{
        $\msg\gets x$
    }
\end{algorithm}}

\newcommand{\hybridT}{%
\removelatexerror
\begin{algorithm}[H]
    \algoTitle{Target}{$\T_\hybrid$}
    \Return{$\N[\loc{D}].\Read()$}
\end{algorithm}}

\newcommand{\secretE}{%
    \removelatexerror
    \begin{algorithm}[H]
    \algoTitle{Evidence}{$\E_{\secret}$}
    \Variable{$\R.x$, $\R.k$}
\end{algorithm}}    
\newcommand{\knownE}{%
    \removelatexerror
    \begin{algorithm}[H]
    \algoTitle{Evidence}{$\E_{\known}$}
    \KwData{$\loc{x}$}
    \Variable{$\R.x$, $\R.k$}
    \OracleR{$N[\loc{x}]$}
    \Assert{$\R.x == \N[\loc{x}].\Read()$}
\end{algorithm}}

\title{Can the Government Compel Decryption? Don't Trust --- Verify}

\author{Aloni Cohen}
\email{aloni@uchicago.edu}
\affiliation{%
  \institution{University of Chicago}
  \city{Chicago}
  \state{IL}
  \country{USA}
}

\author{Sarah Scheffler}
\email{sscheff@princeton.edu}
\affiliation{%
  \institution{Princeton University}
  \city{Princeton}
  \state{NJ}
  \country{USA}
}

\author{Mayank Varia}
\email{varia@bu.edu}
\affiliation{%
  \institution{Boston University}
  \city{Boston}
  \state{MA}
  \country{USA}
  \postcode{02140}
}

\iffull
  \renewcommand\footnotetextcopyrightpermission[1]{} %

\else


\copyrightyear{2022}
\acmYear{2022}
\setcopyright{rightsretained}

\acmConference[CSLAW '22]{Proceedings of the 2022 Symposium on Computer Science and Law}{November 1--2, 2022}{Washington, DC, USA}
\acmBooktitle{Proceedings of the 2022 Symposium on Computer Science and Law (CSLAW '22), November 1--2, 2022, Washington, DC, USA}
\acmDOI{10.1145/3511265.3550441}
\acmISBN{978-1-4503-9234-1/22/11}

\usepackage{etoolbox}
\makeatletter
\patchcmd{\maketitle}{\@copyrightpermission}{
   \begin{minipage}{0.3\columnwidth}
     \href{https://creativecommons.org/licenses/by/4.0/}{\includegraphics[width=0.90\textwidth]{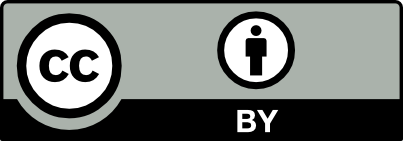}}
   \end{minipage}\hfill
   \begin{minipage}{0.7\columnwidth}
     \href{https://creativecommons.org/licenses/by/4.0/}{This work is licensed under a Creative Commons Attribution International 4.0 License.}
   \end{minipage}
 
   \vspace{5pt}
}{}{}

\makeatother

\fi

\begin{document}

\begin{abstract}
\emph{If a court knows that a respondent knows the password to a device, can the court compel the respondent to enter that password into the device?}
In this work, we propose a new approach to the foregone conclusion doctrine from \emph{Fisher v. U.S.} that governs the answer to this question.
The Holy Grail of this line of work would be a framework for reasoning about whether the testimony implicit in any action is already known to the government.
In this paper we attempt something narrower.
We introduce a framework for specifying actions for which all implicit testimony is, constructively, a foregone conclusion.
Our approach is centered around placing the burden of proof on the government to \emph{demonstrate}
that it is not ``rely[ing] on the truthtelling'' of the respondent.

Building on original legal analysis and using precise computer science formalisms, we propose \emph{demonstrability} as a new central concept for describing compelled acts. We additionally provide a language for whether a compelled action meaningfully \emph{entails} the respondent to perform in a manner that is ``as good as'' the government's desired goal.
Then, we apply our definitions to analyze the compellability of several cryptographic primitives
including decryption, multifactor authentication, commitment schemes, and hash functions.
In particular, our framework reaches a novel conclusion about compelled decryption in the setting that the encryption scheme is deniable:
the government can compel but the respondent is free to use any password of her choice.

\end{abstract}

\maketitle

\section{Introduction}
\label{sec:intro}
In a criminal case,
if the government wants to read the contents of your encrypted phone or laptop, can they compel you to enter or disclose your password? 
If so, when?
The question sounds simple, but it doesn't have an answer.
Instead, it has too many answers.

In the past few years, several state supreme courts in the United States have grappled with this question and reached wildly different rulings.
The Massachusetts Supreme Judicial Court found that you can be compelled to enter your password to decrypt the hard drive contents \cite{commonwealth_v_gelfgatt,commonwealth_v_jones}, whereas the Indiana Supreme Court reached the opposite conclusion \cite{seo_v_state}.
The New Jersey Supreme Court ruled that you can even be forced to disclose your password \cite{state_v_andrews}, whereas the Pennsylvania Supreme Court ruled that passwords themselves cannot be compelled \cite{pennsylvania_v_davis}.
Conflicting decisions also exist among federal circuit courts \cite{duces_tecum,macpro} and
district courts (e.g., \cite{sec_v_huang,in_re_boucher,us_v_kirschner,us_v_fricosu}).

The confusion stems from disagreement about how to adapt the law of pre-digital \emph{act of production} cases, where the government compels production of physical documents (see e.g., \cite{fisher_v_us,us_v_hubbell,us_v_doe}).
How should we overcome this physical-digital divide?
Most prior work starts from the notion that encryption isn't fully understood within the law,
and
reasons about how encryption technology is analogous to---or fundamentally different from%
---physical acts of production.
This article pursues a different path: we use the fact pattern of compelled decryption as a means to interrogate our understanding of the act of production doctrine generally.

\paragraph{The foregone conclusion doctrine}\quad
This work exclusively considers the law in the United States, where the government's ability to compel an action is substantially limited by the Fifth Amendment to the U.S.\ Constitution \cite{const5a}.
The Fifth Amendment prohibits the government from compelling an individual to provide self-incriminating \emph{testimony}, such as oral or written statements that ``disclose the contents of his own mind'' \cite{curcio_v_us}.

Moreover, the government cannot compel an individual---who we call the \emph{respondent}---to perform an action 
which, by its very performance, would reveal \emph{implicit testimony}. 
But there's a catch:
an action that communicates implicit testimony
can be compelled
if all the implicit testimony is already a ``foregone conclusion.'' If the government knows it, the implicit testimony would ``add[] little or nothing to the sum total of the Government's information'' about the respondent's mind \cite{fisher_v_us}, and the government would in no way be ``relying on the truthtelling'' of the respondent.
This idea is called the \emph{foregone conclusion doctrine} and stems from the 1976 Supreme Court case \emph{Fisher v.\ United States} \cite{fisher_v_us}.

\paragraph{Complementing the law with computer science}\quad
This work pursues a formal (in the computer science sense) investigation of the foregone conclusion doctrine.
Our target is a specification of the relevant law, written in the language of math and cryptography,
that captures a meaningful and coherent version of the foregone conclusion doctrine.
Our objective for doing so is to
to gain a deeper understanding of the foregone conclusion doctrine in order to reason clearly and consistently about compelled decryption under varying present and future fact patterns.

We don't want to mechanize the law. Law is flexible and adapts to new situations,
and attempts to encode society's rules into a rigid computer algorithm (aka, ``code as law'') typically don't end well \cite{morrison2020dao,yeung2019regulation}.
Hence,
our pursuit of a computer science specification is not meant to replace the law, but \emph{to illuminate it}.
Courts and legal scholars disagree wildly about compelled decryption. Some of that is due to imperfect analogies. Some is genuine disagreement as to what the law is. Some is simply bad law.
In many cases, it is unclear how the various approaches relate to one another or apply to alternative scenarios.

Our goal is to present our approach in enough detail that one can apply it to unforeseen facts and get a concrete answer. Other scholars can identify both if and where they disagree with our interpretation.
We make no claims that our work is ``the correct interpretation'' of the doctrine. We do claim that it is reasonable and coherent, that it exposes the nuance and ambiguity in the law, and that it serves as a base for others to improve in the future.

\paragraph{Demonstrating foregone conclusions through verification}
Consider a government actor $\G$ and a respondent $\R$.
The overarching question is whether the government can compel the respondent to perform a certain \emph{action} $\A$, such as producing documents or decrypting a computer.
In this work, we specify $\A$ as an algorithm that can request inputs from the respondent's mind, like the document's location or the computer's password.
If this action involves explicit testimony, then it is strictly forbidden and the concept of foregone conclusion is irrelevant. %

Otherwise, to compel this action, the government must generally show two things.\footnote{%
    These two things are neither necessary nor sufficient conditions for compelling an action. For example, they are not sufficient if the act is providing oral or written testimony, and they are not necessary if the act is not incriminating.} 
First, the government must establish that the respondent can perform the action $\A$ by submitting some convincing evidence $\E$ to the court.
Second, in compelling the respondent to perform the action, the government must not ``rely on the truthtelling'' of the respondent $\R$.

A key insight from \emph{Fisher} is that the government does not rely on $\R$'s truthtelling if it has some way of \emph{verifying} that $\R$ performed the required action.
The government's task is typically called ``authentication'' in legal scholarship because when $\A$ is an \emph{act of production} then it corresponds to authenticating in the evidentiary sense. We generalize the concept because
authentication of evidence is a poor fit when $\A$ is an \emph{act of performance} instead; see \S \ref{sec:legal} for details.

Our goal is to make \emph{Fisher}'s verification idea precise.
One approach might be to require that the government specifies a procedure $\V$ that checks whether the respondent $\R$ has performed the action \emph{exactly}.
But this approach is too restrictive. For example, suppose the action is ``enter a password into this password prompt.'' Must $\R$ stand or sit? Eyes open or closed? Can $\R$ make a typo if it succeeds on the second try?
Of course the government doesn't care about any of this -- it just wants the password to be entered.
Another approach might be for the government to specify a set of all acceptable actions, but this set is infinite.
The upshot is that the government's interest lies in the ends, not the means, and we should approach the foregone conclusion doctrine accordingly.

Our insight is to flip $\A$ and $\V$ around.
Rather than asking how to verify a particular action $\A$, in this work we make the verification procedure $\V$ the centerpiece of any foregone conclusion claim.
For example, rather than compelling the respondent to ``enter a password,'' instead the government compels the respondent to ``do something that gets past the login prompt screen and to the point that we can open the file explorer.''
It is then up to the respondent to choose any action $\A$ that \emph{conforms} to this verifier $\V$.

Making the verifier $\V$ the centerpiece of our specification immediately resolves cruel trilemma of ``relying on the truthtelling.''
The respondent is free to perform
any
action it wants so long as $\V$ accepts; there is no longer a question of honesty or evasion.
The only remaining question is how to evaluate whether the respondent's task is feasible in the first place.
We impose the weakest requirement: the government must demonstrate that, based on the evidence, the respondent can perform some action $\xA$ that $\V$ accepts.
We call $\xA$ the \emph{exemplar action}, and we call $\V$ \emph{demonstrable} if it has such an exemplar action.
In summary, our shift in focus from $\A$ to $\V$ changes how we view compulsion:
instead of ``is action $\A$ a foregone conclusion?'' we ask ``is verifier $\V$ \emph{demonstrable}?''

\paragraph{Meaningfulness}
The cost of this shift in focus is a lack of clarity about what the government will get.
Presumably the government has a target action $\T$ in mind that it really wants the respondent to perform.
There are many demonstrable $\V$'s for any given scenario; for example, the trivial $\V$ that accepts any action is always demonstrable,
but it provides no guarantee that the government will receive the result of $\T$ or anything remotely resembling it.
Another example of a demonstrable-but-meaningless verifier is to compel the respondent to flip a coin and verify that the result is one of ``heads'' or ``tails'' (rather than, say, ``42'').

But sometimes, it may be possible to gather enough evidence and specify $\V$ in enough detail that the respondent must
perform an action that is more-or-less equivalent to $\T$.
We call this property \emph{entailment} and we codify the claim of ``more-or-less equivalent'' using the cryptographic notion of extraction.
That is, $\V$ entails $\T$ if the government can transform the result of any conforming action into the fruits of the action $\T$ it desired all along.

\paragraph{Our contributions}
This work makes several contributions in the domains of law and computer science.
\begin{compactitem}
\item As a matter of law, we adopt a view of implicit testimony as comprised of dual but equally important parts: ex ante \emph{ability} and ex post \emph{conformity}. 
We also present a new verification-centric approach to compelling acts while constructively satisfying the foregone conclusion doctrine (\S \ref{sec:legal}).

\item We define the concepts of demonstrability and entailment within a formal computer science framework (\S \ref{sec:formalism}-\ref{sec:entailment}). The framework is expressive enough to capture, for example, the government's uncertainty as to whether an encryption scheme is deniable or not.

\item Connecting the computer science to the law, we investigate whether several acts involving cryptography are compellable within our framework (\S \ref{sec:examples} and Appendix \ref{app:examples}).
\end{compactitem}
Generally speaking, we believe that \S \ref{sec:legal} will be of most interest to readers with a background in law, and \S \ref{sec:formalism}-\ref{sec:entailment} will appeal to a computer science audience.
That said, we highlight some of the takeaways about the compellability of cryptography, for both audiences.
\begin{compactitem}
\item Decryption by entering a password is entailable in the (typical) case that the encryption scheme is committing. This is true even when using two-factor authentication. More generally, opening \emph{commitments} is generally entailable. 

\item Our framework reaches a novel middle ground conclusion if decryption is deniable -- meaning that it will reveal different files when given a duress password. 

\item In contrast to decryption and opening commitments, the acts of encrypting and creating commitments are not entailable.

\item Hash functions add a twist: compelling preimages is demonstrable and---under a cryptographic assumption---entailable.
\end{compactitem}

\paragraph{Comparison to prior work}
This article is inspired by several recent works that codify legal principles with computer science
techniques \cite{EC:GarGolVas20,FORC:CohenDMS21,PNAS:CohenN20,arXiv:JudFei22,PODS:Nissim21,Nissim2017bridging,Altman2021hybrid}, including Scheffler and Varia's related work on compelled decryption \cite{USENIX:SchVar21}
that we contrast with in \S \ref{ssec:usenix}.
Additionally, this work is influenced by a large body of legal scholarship on compelled decryption
\cite{kiok,sacharoff,winkler,cohen-park,kerr2016,kerr2018compelled,kerr2020,mcgregor,terzian2013fifth,terzian-dec},
some of which suggest broad powers to compel decryption and some argue that the government can only compel files that they already know with ``reasonable particularity.''

Our framework generally aligns with the broader interpretation, and it considers all acts of performance beyond just the produced files.
In the next section, we delve deeper into case law, scholarship, and the value of a verification-centric approach as a matter of law.

The question of deniable encryption highlights the differences between the current work and prior work (see Section~\ref{sec:enter-password:deniable}).
Deniable encryption introduces a \emph{duress} password which, when entered, allows a device to be decrypted to something other than its true contents.
In \emph{Comm.\ v Gelfgatt}, the court simply ordered the respondent to use the true password rather than the duress password.\footnote{\label{note:gelfgatt}{\emph{Comm.\ v Gelfgatt}~\cite{commonwealth_v_gelfgatt} at note 10 (the decryption ``protocol'' states in relevant part: ``The defendant shall manually enter the password or key to each respective digital storage device in sequence. \dots The defendant is expressly ordered not to enter a false or 'fake' password or key'').}}
Sacharoff sets aside this ``niche case because deniable encryption remains rare'' \cite{sacharoff}. Kerr argues that duress passwords are ``unlikely to raise significant Fifth Amendment issues'' because, if used, ``the government will not realize what the user has done.''\footnotemark~Cohen and Park argue that the government would be unable to compel decryption without some way of distinguishing the respondent's use of the true and duress passwords \cite{cohen-park}.

    \footnotetext{{\emph{Kerr}~\cite{kerr2018compelled} at note 78. It is not clear to what extent these are statements of legal doctrine (e.g., no implicit testimony at issue; the testimony is always a foregone conclusion) or of practical reality (e.g., the deniable encryption will be invisible and therefore will not be at issue in the proceedings). 
    In one sense, Kerr agrees with us:  when ordered to ``bypass [the] password prompt'' a respondent is free to use a duress password. But he seems to suggest that this is merely a de facto freedom, and that use of a duress password could ``violate the Decryption Order.'' Under this reading, an order forbidding the use of the duress password (as in \emph{Gelfgatt}, note~\ref{note:gelfgatt} above)  may be permissible even if the government ``will not realize what the user has done.''
    In contrast, this paper argues that the decryption order can only forbid the use of a duress password if the government is able to verify whether or not the true password was used. Otherwise, the act of decrypting with the true password would communicate testimony as to conformity that is not a foregone conclusion.}}

In contrast, our framework would allow the government to compel the respondent enter one of the true or duress passwords, but the respondent would be free to choose. If the government also had evidence as the contents of the encrypted drive---say, specific filenames as in \textit{In re Boucher} \cite{in_re_boucher}---the government could require that the apparently-decrypted device contain the known files by checking for those files in the verification procedure. In such cases, our framework yields a procedure more like Sacharoff's and Cohen-Park's than Kerr's. But like Kerr, meaningful acts are compellable even without any specific knowledge of the device's contents.

\section{A new take on foregone conclusions}
\label{sec:legal}

This section gives a too-brief overview of the law of the foregone conclusion doctrine and of our new approach to that doctrine. A full treatment is out of scope in this paper.\footnote{A draft of a companion law review article was presented at the Privacy Law Scholars Workshop in 2021 \cite{plsc}.}

The Fifth Amendment to the US Constitution provides individuals a privilege against being ``compelled in any criminal case to be a witness against himself.''\footnote{\textsc{U.S. Const.} Amend. V.}
Courts and legal scholars disagree about how this privilege applies to so-called compelled decryption cases, wherein the government seeks to compel an individual to unlock or decrypt an encrypted phone, computer, or hard drive. At issue is how the foregone conclusion test applies to the act of decryption. 

The Fifth Amendment's privilege applies only to acts that are ``testimonial.''
\footnote{Fisher v. United States, 425 U.S. 391, 408 (1976) 
\cite{fisher_v_us}}
An act is testimonial when it ``explicitly or implicitly[] relate[s] a factual assertion or disclose[s] information.''
\footnote{United States v Doe, 487 U.S. 201, 208 (1988) 
\cite{doe_v_us_ii}}
As such, a person may be compelled to furnish a blood sample 
\footnote{Schmerber v. California, 384 U.S. 757, 763-764 (1966) 
\cite{schmerber_v_california}} or to 
put on a shirt
\footnote{Holt v. United States, 218 U.S. 245 (1910) 
\cite{holt_v_us}} no matter how incriminating. But a person may not be compelled to make incriminating oral or written statements.

\paragraph{Implicit testimony and the foregone conclusion doctrine}\quad
An act may have "testimonial aspects" even if it does not involve any oral or written testimony as such. \emph{Fisher v. U.S.} lays out a version of this idea now called the act of production doctrine.
In \emph{Fisher}, the government issued a subpoena to compel the defendant to produce certain documents. ``The act of producing the documents has communicative aspects of its own, wholly aside from the contents of the papers produced. Compliance with the subpoena tacitly concedes the existence of the papers demanded and their possession or control by the [respondent]. It also would indicate the [respondent]'s belief that the papers are those described in the subpoena.''
The lesson is that testimony may be implicit in an act. \emph{Fisher}'s  examples---existence, possession, and authenticity---are an application of this idea to the specific act of producing physical evidence. %

\emph{Fisher} also introduced the \emph{foregone conclusion doctrine} which provides a way for the government to compel an act of production despite the implicit testimony.
Kerr gives a typical account of the doctrine \cite{kerr2018compelled}:
    \begin{quote}
    [W]hen the testimonial aspect of a compelled act ``adds little or nothing to the sum total of the Government’s information,'' any implied testimony is a ``foregone conclusion'' and compelling it does not violate the Fifth Amendment. . . .
    The Court [in \emph{Fisher}] held that the testimony implicit in handing over the tax documents was a foregone conclusion because the government was ``in no way relying on the 'truthtelling' of the [respondent]'' to prove it.
    \end{quote}

In other words, ``A valid privilege exists only when the compelled act is testimonial under the act of production doctrine but is not a foregone conclusion'' \cite{kerr2018compelled}.
Importantly, the foregone conclusion doctrine only considers testimony implicit in an act. It does not apply to `pure' oral or written testimony -- no matter how foregone, pure testimony cannot be compelled.\footnote{\emph{Fisher}, 258 U.S. \cite{fisher_v_us} (``It is doubtful that implicitly admitting the existence and possession of the papers rises to the level of testimony within the protection of the Fifth Amendment'' at 411; additionally ``it [the subpoena to produce documents] does not compel oral testimony, nor would it ordinarily compel the taxpayer to restate, repeat, or affirm the truth of the contents of the documents sought[, t]herefore the Fifth Amendment would not be violated'' at 409.  These statements contrast this foregone conclusion doctrine against \emph{explicit} compelled admittance of the existence and possession of the papers by oral testimony.  Thus the foregone conclusion doctrine may compel only implicit testimony and not explicit oral testimony, even if that information is already known to the government.)}
And the foregone conclusion doctrine is unconcerned  with information that may result from the act, but is not implied by the act. Borrowing an example from Kerr, the testimony implicit in the act of opening a door to a treasure filled room does not include the treasure itself, only the ``door-opening'' testimony. The value of the treasure is not relevant to the Fifth Amendment privilege.

\paragraph{Ability and conformity}\quad
Applying the foregone conclusion doctrine first requires answering the question: what are the implicit testimonial aspects in the act of decryption (or some other act of performance) \cite{duces_tecum, kerr2018compelled, cohen-park}?
This may sound easy, but it is where many compelled decryption analyses diverge.
For example, Kerr concludes that `` `I know the password’ is the only assertion implicit in unlocking the device'' \cite{kerr2018compelled}. Cohen and Park respond that ``[d]eniable encryption complicates the very notion of `the' password by introducing an alternate duress password'' that appears to decrypt but hides certain files \cite{cohen-park}.

We adopt the view that there are two types of testimony implicit in acts of performance: ability and conformity. By honestly performing an act specified by the government, a respondent implicitly asserts ``I can do it'' and ``I did do it.''  \emph{Ability}---I can do it---is the respondent's ex ante belief that she is able to perform the compelled act. \emph{Conformity}---I did do it---is the respondent's ex post belief that she did indeed perform the compelled act as specified.
This view has appeared in prior works, though the terms are new. For example, Kerr describes these as the ``two kinds of beliefs'' implicit in a compelled act.\footnote{See Kerr~\cite{kerr2018compelled} at I.A.}

The importance of ability to the foregone conclusion doctrine is well accepted, and the notion of ability has evolved beyond its act-of-production genesis.
For physical production of some object, ability amounts to the existence of the object and the respondent's possession and control thereof, \emph{Fisher} \cite{fisher_v_us}.
For entering a password, ability often amounts to the respondent's knowledge of the password \cite{kerr2018compelled}.
For decryption and production of encrypted drives, ability communicates the respondent's ``possession, control, and access to the encrypted portions of the drives; and of his capability to decrypt the files'' \cite{duces_tecum}.

By contrast, conformity's role is often minimized or overlooked.
\emph{Fisher} and all subsequent caselaw and scholarship focus on only one aspect of conformity:  \emph{authenticity}.
In Kerr's analysis, only ability informs his analysis of compelled decryption.\footnote{See Kerr~\cite{kerr2018compelled} at II.}
An act of production in response to a subpoena implicitly communicates ``that the articles produced are the ones demanded'' \cite{wigmore}.
But while authenticity --- whether something is what it purports to be --- is appropriate for acts of production, it is an ill fit for acts of \emph{performance} where the respondent must {do} something instead of {produce} something. 
Compelled decryption---say, by typing a password into an encrypted computer---is typically best treated as an of performance, not production.\footnote{See Cohen and Park \cite{cohen-park}, especially Sections~ IV.C, IV.D, V.B, V.C, and VII.}
We consider conformity as a generalization of authenticity that is relevant to acts of performance as well as production,
and our work gives conformity equal consideration for the first time.

A key difference between how Kerr and Cohen-Park treat compelled password entry is the importance they attach to conformity \cite{kerr2018compelled,cohen-park}. Both agree on ability: the act of entering a password communicates knowledge of the password. But they disagree about conformity in the context of deniable encryption. Kerr dismisses it, stating that ``the government will not realize what the user has done.'' Cohen and Park recognize that deniable encryption implicates conformity exactly because the government cannot tell what the user has done. 

\paragraph{The foregone conclusion exception}
\quad
Once the implicit testimonial aspects of the act of decryption are identified, the foregone conclusion doctrine asks: Are those aspects foregone conclusions \cite{duces_tecum, kerr2018compelled, cohen-park}?
Ability is a foregone conclusion if the government can demonstrate that the respondent can perform the action ex ante. Conformity is a foregone conclusion if the government can demonstrate ex ante that it will be able to verify whether the respondent performed the action ex post.

Ability raises no new difficulties. The government must show that the respondent can perform the act. Depending on the act, this amounts to showing that the papers exist and and are in his control \cite{fisher_v_us, us_v_hubbell}; that he knows the password (e.g. \cite{pennsylvania_v_davis,commonwealth_v_gelfgatt,commonwealth_v_jones,sec_v_huang,us_v_burns,us_v_maffei}); or that he is able to speak, write, or make a gesture \cite{schmerber_v_california}. In our formalism, the government meets its ability burden by (1) describing an \emph{exemplar action} that is  ``conforming'' (see below), and (2) presenting evidence that the respondent is able to perform the exemplar action.

Conformity is more challenging. Conformity is a foregone conclusion if the government can demonstrate that it will be able to \emph{verify} whether or not the respondent performed the compelled act after the fact. Our framework requires the government to specify the procedure it will use to verify the act, and deems \emph{any} act that successfully verifies as conforming.

Under this approach, \textbf{the verification procedure is the centerpiece of a motion to compel}.
The question of conformity is resolved by construction: any action that verifies, verifies. To show that ability is a foregone conclusion, the government must demonstrate that there exists a conforming exemplar action that the respondent can perform. If so, we say the verification procedure is \emph{demonstrable}.

Our approach to conformity addresses a shortcoming in the status quo's treatment of conformity in the guise of authenticity.
In compelled decryption cases to date, courts have generally dealt with authenticity in one of four unsatisfactory ways.
{(In each, replacing ``authenticity'' with ``conformity'' would not fix the problem.)}
First, ignore the issue altogether: ``the authenticity element is routinely cited but only applied loosely if at all.''\footnote{In re Search of a Residence in Aptos, California 95003, No.17-mj-70656-JSC-1, 2018 WL 1400401 \cite{in_re_search_of_residence_in_aptos_ca}, at 10.}
Second, defer the question to trial, which we fear could allow the government to make improper use of fruit of a poisonous tree.
Third, hold that decryption is self-authenticating---a technological assertion that deniable encryption, for example, disproves.
Fourth, attempt to reason about the {production} of some \emph{thing}---a password, plaintext, or ciphertext ---{rather than {the performance} of an act,} a confusing and unworkable approach.\footnote{See also Cohen and Park \cite{cohen-park} at VII.}

Conformity should not be a mere afterthought. Even before \emph{Fisher}, Wigmore's \emph{Treatise on Evidence} \cite{wigmore} states that production of documents or chattels may only be compelled from a respondent 
``without the use against him of process relying on his truthtelling.''
Where conformity is not a foregone conclusion, the government would be improperly ``relying on his truthtelling.''
This reliance is clear in \emph{Comm.\ v.\ Gelgatt}, where the court simply ordered the respondent ``[n]ot to enter a false or ‘fake’ password'' although it could not check compliance.\footnote{Commonwealth v. Gelfgatt, 11 N.E.3d 605 (Mass. 2014)  \cite{commonwealth_v_gelfgatt}, (see footnote 10, emphasis removed)}
Such orders revive the cruel trilemma that is at the heart of the religious tribunals that gave rise to the right to silence: the respondent must choose among perjury, contempt of court, and self-incrimination.

\paragraph{A new procedure for motions to compel.}\quad
Based on our framework, we suggest a new procedure for motions to compel decryption or other acts of performance.
First, the government makes a motion to compel. It specifies a  verification procedure and demonstrates that the respondent is able to perform some exemplar action that verifies. 
To do so, the government must clearly detail the steps of both the verification procedure and the exemplar action in a manner understandable by the judge and the respondent. The exemplar action must be one that the government would have been able to compel if all implicit testimony was forgone. In particular, the exemplar cannot require explicit oral or written testimony.
The government must also introduce any evidence that it needs to demonstrate the respondent's ability to perform the exemplar.

Second, the respondent may challenge the motion to compel in a few different ways. She may challenge the accuracy of the government's evidence. She may also argue that she cannot perform the exemplar action. Or she may argue that the Fifth Amendment privilege precludes the exemplar action for reasons other than implicit testimony. 
(The respondent may \emph{not} appeal to conformity, which is foregone by construction.)

The court must then decide whether to quash the motion or to grant it by issuing an order. The order would compel the respondent to perform an action of his choosing that passes the government-specified verification procedure. The order would contain within it a complete description of the procedure and the exemplar action.

In the new procedure based on our framework, it is important for the order to clearly specify the evidence, method of verification, and exemplar.
Doing otherwise would undermine the Fifth Amendment privilege. The resulting uncertainty would present the respondent with a cruel choice: either risk contempt for failing the verification procedure or risk disclosing implicit testimony which is \emph{not} foregone.
Moreover, the idea that the government must provide the verification procedure before compelling an action is already present in, for example, \emph{U.S. v Bright}.\footnote{United States v. Bright, 596 F. 3d 683, 693 (9th Cir. 2010) \cite{us_v_bright} (the government does not ``need to prove that it had previously authenticated the same documents \ldots it need[s] to show only that it \emph{could} do so'' (emphasis added)).} 
As an added benefit, a well-defined verification procedure might be relevant in disputes over the introduction or authentication of evidence at trial time.

Under our framework, the respondent can perform \emph{any} act that passes the government's verification procedure.
Presumably the government has in mind some \emph{target action} that it wants the respondent to perform (e.g., the exemplar). But if the government cannot tailor its verification procedure appropriately, then the compelled act may be meaningless. If however the verification is tailored so as to require something ``just as good'' as the target, we say that the verification procedure \emph{entails} the target.

For example, consider the verification procedures the government could use to compel decryption of a device encrypted using deniable encryption.  It could, of course, check that the device appears to be logged in.  However, compelling this would clearly allow a decryption under both the true and duress passwords to pass the verification procedure. 
Using the duress password is allowed---the respondent would not be acting untruthfully. 
But if the government had specific knowledge of a file visible only using the true password, then checking for that file would prevent the duress password's use.

\section{Computational formalism}
\label{sec:formalism}

In this section, we describe the formal model of computation that we consider in this work. It is inspired by 
the computation and communication model within Canetti's universal composability (UC) framework \cite{FOCS:Canetti01} and it directly extends the computational model used within Scheffler and Varia's prior work on compelled actions and foregone conclusions \cite[Appendix A]{USENIX:SchVar21}.

Concretely, we model computation as a collection of stateful, event-driven interactive Turing machines (ITM).
These Turing machines have two tapes called a method tape and an input tape.
When a machine is invoked with a string $\methodstyle{Method}$ on its method tape and an input $\mathit{input}$ on its input tape, it executes the code of $\methodstyle{Method}$.
Additionally, the machines have additional tapes that allow them to communicate with each other using authenticated, confidential channels, following the model of Canetti, Cohen, and Lindell \cite{C:CanCohLin15}.

In this work, we describe the methods of an interactive Turing machine using object-oriented pseudocode. Here is an example:

\exampleITM
Our ITMs have variables and methods. By default, all variables are private and local and all methods are public. The example machine $M$ above has three local variables, where $x_1$ can be initialized arbitrarily, $x_2$ is initialized to some nonzero value, and $x_3$ is initialized to the value $5$.
$M$ also has public methods $M.\methodstyle{Set}$ and $M.\methodstyle{Send}$.
So, we can invoke machine $M$ by placing $\methodstyle{Set}$ on its method tape and $(i, x')$ on its input tape.
(When an ITM has only one method and no variables, we omit naming the method.)
The $\methodstyle{Send}$ command sends a communication to another Turing machine $M'$ that it should invoke a method of its own called $\methodstyle{Receive}$.\footnote{Though the UC model guides us, a complete UC specification is premature at this stage in this line of work and would make this paper inaccessible to its intended audience, most of whom are unfamiliar with the UC model. For the same reasons, we do not attempt a precise compilation from our pseudocode machines to ITMs.}

\label{ssec:partial-spec}
Looking ahead, in this work it will be useful to model situations in which the government knows that a device has a specific behavior in some situations, but may not know the entire code of the device. For this reason, we define a partial ordering $\prec$ on machines.

\begin{definition}
Let $M_1$ and $M_2$ be interactive Turing machines, and let $S$ denote the set of all methods specified within $M_1$.
We say that $M_1$ is a \emph{partial specification} of $M_2$ (equivalently, $M_2$ \emph{implements} $M_1$), denoted $M_1 \prec M_2$, if:
\begin{compactitem}
\item $M_2$ contains at least all of the methods  specified in $M_1$, i.e., the methods specified in $M_2$ form a superset of $S$, and

\item For all executions of the machines that only invoke the methods within $S$, if $M_1$ halts with some output then $M_2$ halts with the same output.
This semantic equivalence must hold even over multiple, sequential invocations of the stateful machines $M_1$ and $M_2$.
\end{compactitem}
$M_1$ \emph{fully specifies} $M_2$, denoted $M_1\sim M_2$ if $M_1 \prec M_2$ and $M_2 \prec M_1$.
\end{definition}

In our example above, we can consider $M \prec M''$ for any ITM $M''$ that contains $\methodstyle{Set}$ and $\methodstyle{Send}$ commands with the same code as listed above. That said, $M''$ might have additional commands, such as a $\methodstyle{Delete}$ method that sets all variables to zero or a 
$\methodstyle{Write}$ method that covertly overwrites a private variable.
Note that $\prec$ is uncomputable in general. 

In this work, machines provided by the government (i.e., $\V$, $\xA$, $\T$, and $\P$ defined below) tend to be
fully specified.
On the other hand, 
machines controlled by others (i.e.,
$\R$ and 
$\N$) are typically only partially specified since the government
may only have some incomplete evidence
about how they operate.

\section{A verification-centric approach to compulsion}
\label{sec:verification}

We consider a \emph{government} actor \G, a \emph{respondent} \R, and the outside world (\emph{nature}) \N. 
\N is an infinite
collection of ITMs $\N[0]$, $\N[1]$, $\N[2]$, $\dots$, where the index is called the \emph{location}. 
Oftentimes we will specify a device $D$ as either a partial or full representation of an oracle at (for example) location $\ell$ in Nature, denoted as $D\prec \N[\ell]$ or $D \sim \N[\ell]$ respectively.
Some of nature's ITMs simply store information and their only functionality is a $\Read()$ method that returns that information.  We call such locations \emph{read only}.
In these cases, we sometimes simplify the notation by writing $x \gets \N[1]$ instead of $x \gets \N[1].\Read()$, for example.

\G and \R don't interact with each other or with \N directly. Instead, they each output ITMs that specify how they choose to act. The respondent outputs the \emph{action} $\A$ it will take, and the government outputs a \emph{verifier} $\V$ that represents how it will authenticate the action performed by the respondent.
One can think of \V and \A as the actions \G and \R choose to perform under the circumstances.
Separating \V and \A from \G and \R allows us to describe and reason about the interaction between government and respondent (and nature) while assuming very little about \G and \R.

We denote an execution $\exec{\V^\N}{\A^{\N,\R}}$.\footnote{%
    We borrow this notation from oracle machines to reinforce the ways the \V, \A, \N, and \R may and may not interact. More precisely though, $\V$ can invoke $\N$ and $\A$ can invoke $\N$ and $\R$ as (global) subroutines, in the UC language of Canetti \cite{FOCS:Canetti01}.}
Here, \V and \A are ITMs that may interact freely with one another and with \N. However, $\R$ may only interact with $\A$, not with \N nor \A. 
The output of the execution is \V's output: one of $\Accept$ or $\Reject$ along with a transcript $\tau$ of its view and execution.

\begin{definition}[$\V$ accepts $\A$]
We say $\V$ \emph{accepts} $\A$ \emph{with respect to} $\R,\N$ if $\langle \V^\N, \A^{\N,\R} \rangle$ returns $\Accept$ with probability 1 over the coins of $\V$, $\A$, $\R$, and machines in $\N$.
If $\N$ and $\R$ are clear from context we abbreviate this as $\V$ \emph{accepts} $\A$.
\end{definition}

\subsection{Evidence}
\label{ssec:evidence}

We envision the government as having some \emph{evidence} that certain states of the world ($\N$) or the respondent's mind ($\R$) are or are not true.  We represent the evidence as a binary relation $\E : (\R, \N) \mapsto \{0,1\}$. 
$\E(\R,\N)=1$ means that the pair $(\R,\N)$ is \emph{consistent} with the evidence. Even if $\E(\R,\N)=1$, the ``true'' state of the world can be some other $(\R',\N')$.
Conversely, $\E(\R,\N) = 0$ means that $(\R,\N)$ is inconsistent with the evidence. Namely, the evidence refutes that setting of $\R$ and $\N$. (Note that \E is generally uncomputable.)

\begin{definition}[$\E$-consistency]
Let $\E$ be a binary relation.
We say that $\R,\N$ are \emph{$\E$-consistent} if $\E(\R,\N)=1$.
\end{definition}

In our framework, $\E$ is exogenous and correct. We can imagine that the government has done the work to collect the evidence, introduced it in the court, and demonstrated its truth.
As such, the ``true'' setting of $\R,\N$ is $\E$-consistent.  (As a corollary, $\{(\R,\N) : \E(\R,\N) = 1 \}\neq \emptyset$.)  We do not consider incorrect evidence.

It will be useful to reason about cases when the government has greater or lesser evidence. We define a (weak) partial ordering over evidence which captures whether evidence $\E_2$ is ``at least as strong as'' evidence $\E_1$. This can arise when the government goes and gathers more evidence, or when it introduces additional evidence that it already had.

\begin{definition}[Partial ordering of evidence]
Let $\E_1$ and $\E_2$ be binary relations, and let $\C_{\E_1}, \C_{\E_2}$ be the set of $\E_1$- and $\E_2$-consistent worlds respectively. That is, $\C_{\E_i}$ is the set of $(\R,\N)$ for which $\E_i(\R,\N)$ outputs 1 (note that these sets may be uncomputable).

We say $\E_2 \succeq \E_1$, or \emph{$\E_2$ has at least as much evidence as $\E_1$}, if $\C_{\E_2} \subseteq \C_{\E_1}$.
Equivalently, $\E_2$-consistency implies $\E_1$-consistency.
\end{definition}

As discussed in Section \ref{ssec:partial-spec}, most of the time we consider partial specifications of ITMs.
In Section
\ref{sec:examples}
we will sometimes find it useful to consider evidences where the machines are specified fully.  We define the following useful shorthand to capture this scenario.

\begin{definition}[Fully specified $\E_D$]
\label{def:E-D}
For evidence $\E$ and ITM $D$, we define $\E_D\succeq \E$ to be the evidence where every assertion of the form ``$D \preceq M$'' is replaced with ``$D \sim M$.''

\end{definition}

\subsection{Demonstrability}

Our key mechanism to avoid ``relying on the truthtelling'' is for the government to specify a verifier $\V$ which will be used to verify the respondent's response.
There must be some \emph{exemplar action} $\xA$ which demonstrates that the respondent is \emph{able} to make the verifier accept, however in the context of the court case the respondent is allowed to perform \emph{any} action $\rA$ that results in the verifier outputting $\Accept$.

\begin{definition}[Demonstrability]
\label{defn:demonstrability}
A verifier $\V$ is \emph{demonstrable} with respect to evidence $\E$ if
there exists an ``efficient'' (see Remark \ref{rem:efficiency})  action $\xA$ such that for all $(\R,\N)$ that are $\E$-consistent:
\begin{compactenum}
\item Every method call by $\xA$ to $\R$ produces some output.
\item $\V$ accepts $\xA$ with respect to $\R,\N$
\end{compactenum}
\end{definition}

The first requirement essentially requires that $\R$ is able to perform any actions in $\xA$ that directly involve the respondent's mind.  Not only does this requirement ensure that $\E$ contains some evidence that $\R$ has the method called, the fact that the method must produce some output for \emph{all} settings of $\R$ ensures that $\E$ must specify some behavior on the output of the method (or else the $\R$ for which the method outputs nothing will violate the statement).

Our envisioned courtroom procedure for a motion to compel would require $\G$ to provide $\V$, $\xA$, and $\E$  to the court and to $\R$ (see Section \S \ref{sec:legal}).

\begin{remark}
\label{rem:efficiency}
We deliberately leave the ``efficiency'' requirement of $\xA$ somewhat vague in Definition \ref{defn:demonstrability}.  We expect that the court will want to define ``ability to perform $\xA$'' in a way that implies some kind of reasonable time and effort limit.  For example, if $\xA$ instructed $\R$ to decrypt a ciphertext, this is \emph{theoretically} doable even without the key -- if the respondent spends time greater than the lifetime of the universe running a brute-force decryption algorithm.  Instead, we expect a more typical setting is when $\G$ knows $\R$ is capable of decrypting the ciphertext quickly (say, by using the key).
We leave incorporating computational requirements to future work. 
\end{remark}

The respondent may perform any action $\rA$ that results in the verifier outputting $\Accept$.  We call such actions  \emph{conforming}:

\begin{definition}[Conformity]
$\A$ \emph{conforms with} $\V$ for a given $\N, \R$ if $\V$ accepts $\A$ with respect to $\N,\R$.
If any of these algorithms are randomized, then acceptance must hold with probability 1 over all random choices.
We also say $\A$ is $\V$-\emph{conforming}.
\end{definition}

Lemma \ref{lem:demonstrability-monotonic} shows that the more evidence the government reveals, the more it can compel. An action never becomes ``less compellable'' by revealing more evidence. 
Focusing narrowly on compulsion, it is in the government's best interests to reveal all evidence it possesses.
We also remark that adding more evidence only restricts the respondent's possible response actions.

\begin{lemma}[Demonstrability, conforming are monotonic in $\E$]
\label{lem:demonstrability-monotonic}
Let $\V$ be demonstrable with respect to evidence $\E_1$, and let $\xA$ be $\V$-conforming.
For all $\E_2 \succeq \E_1$, $\V$ is  demonstrable with respect to $\E_2$, and $\xA$ is $\V$-conforming.
\end{lemma}

\begin{proof}
Let $\C_{\E_i}$ be the set of $\E_i$-consistent $(\R,\N)$ for $i \in \{1,2\}$.
By definition of demonstrability, $\forall (\R,\N) \in \C_{\E_1}$, the execution $\langle \V^\N, \xA^{\N,\R} \rangle$ returns 1. 
$\E_2 \succeq \E_1$ implies that 
$\C_{\E_2} \subseteq \C_{\E_1}$. Hence, $\forall (\R,\N) \in \C_{\E_2}$, so the execution $\langle \V^\N, \xA^{\N,\R} \rangle$ returns 1. 
Thus, $\V$ is also demonstrable with respect to $\E_2$, with exemplar $\xA$. 
\end{proof}

\section{Entailment}
\label{sec:entailment}

In the last section, we described how, 
in order to avoid relying on the respondent's truthtelling, the government may provide a demonstrable verifier $\V$, along with an exemplar action $\xA$ and evidence $\E$, and compel the respondent to take some action $\rA$ that is accepted by $\V$.
Ultimately though,
the government likely has the goal of compelling the respondent to perform some specific \emph{target action} $\T$.
In this section, we ask: is there some demonstrable $\V$ that forces $\R$ to perform the target action $\T$, or something ``just as good''?

To illustrate why this is a different property than demonstrability, consider that there are many demonstrable $\V$'s for any given scenario. If nothing else, the trivial $\V$ that always outputs \Accept is demonstrable for any exemplar action and evidence. With this trivial $\V$, even though the government might hope that the respondent would choose to perform the exemplar action $\xA$, the respondent could take any action $\A$ whatsoever in order to satisfy $\V$. But in some scenarios, it may be possible to specify $\V$ in enough detail that $\R$ has no choice but to perform an action that is more or less equivalent to $\T$.

We call this property \emph{entailment}, and we say that $\V$ entails $\T$ if $\G$ can recover the fruits of $\T$ after $\R$ performs \emph{any} $\V$-conforming action $\A$ (including but not limited to the exemplar action $\xA$). Namely, the government has some method $\P$ that ``post-processes'' the result of $\A$ to recover what would have resulted if $\R$ had run $\T$ itself. 

\begin{definition}[Entailment]
\label{defn:entailment}
Consider an oracle TM $\T$ (target) that takes no input and outputs a string in $\{0,1\}^\ast$.
$\V$ \emph{entails} $\T$ with respect to evidence $\E$ if there exists an efficient oracle TM $\P$
such that for all $\E$-consistent $\N$, $\R$ and all $\V$-conforming $\A$: $$\P^{\N'}(\trans) \equiv \T^{\N,\R}(),$$
where $\trans \gets \exec{ \V^{\N}}{\A^{\N,\R}}$ is the transcript of the execution between $\V^{\N}$ and $\A^{\N,\R}$, and $\N'$ is the state of nature after the interaction between $\V$ and $\A$.
We say $\T$ is \emph{entailable} if there exists demonstrable $\V$ that entails $\T$.
\end{definition}

To complete this definition, it only remains to define what $\equiv$ means.
We begin with the easier case: if all Turing machines in the definition are deterministic, then $\equiv$ denotes exact equality.
That is: the Turing machines $\P$ and $\T$ must output the same string, even if they follow very different paths to get there.
The harder case is when randomness is involved (formally, when the Turing machines have randomness tapes following the model of Canetti \cite{FOCS:Canetti01}).
In this case we impose a high bar for entailment: for each setting of the randomness tapes of $\P$, $\V$, $\A$, $\T$, $\R$, and all of the machines within $\N$, it must be the case that the outputs of $\P$ and $\T$ are identical.

\begin{remark}
The requirement that $\P$ is ``efficient'' is necessary for entailment to be meaningful.
For example, consider the case of compelled decryption.
If the post-processor could run for an unbounded amount of time, then $\P$ could simply ignore the respondent's action and brute-force the decryption, allowing the government to entail the act of decryption even when the respondent did not have the key!
For this reason, Definition \ref{defn:entailment} requires that the post-processor $\P$ be efficient.
As in Remark~\ref{rem:efficiency}, we leave the efficiency requirement underspecified and for future work.
\end{remark}

\begin{remark}
We note that our definition of entailment only considers target actions $\T$ that produce some output. As such, it does not require
that the state of nature $\N_{A,P}$ that results after $\A$ and $\P$ is the same as the state of nature $\N_{T}$ that would have resulted if $\R$ performed $\T$. 
But entailment is expressive enough to capture a somewhat weaker guarantee. Suppose there is a function $q^\N$ that reads from $\N$ whatever the government believes ex ante is relevant, and outputs the result. If $\V$ entails $\T_q \triangleq q \circ T$, then $\N_{A,P}$ and $\N_{T}$ will be the same in the (ex ante) relevant part.
\end{remark}

\subsection{Impossibility of entailing unknown goals}

In this section, we show that it should not be possible to entail actions where the government does not even know what property they would need to check to verify the truth.
We illustrate %
this concept by example and then show a general theorem.

Suppose the evidence states that $\R$ has some secret $z$ (for instance, a list of places she had been that week).  And suppose $\R$ has a method $x()$ that can state one valid value for that statement (e.g. where she was last night).
Can the government compel $\R$ to reveal where she was sometime this past week?  Or, in our parlance, is the target action which returns $\R.x()$ entailable?
(For the sake of this example, suppose the government was permitted to compel \emph{any} act they could verify, even explicit verbal testimony.)

If the government has no evidence that the respondent was located anywhere in particular in the last week, then intuitively, there should be no way to entail the action of $\R.x()$.\footnote{Observe that there are demonstrable $\V$s with ``return $\R.x()$'' as the exemplar. However, the respondent may respond by, for example, responding to the query with ``Boston'' even if she had never been there. This is not the result of $\R.x()$, but the verifier must accept it anyway because the government has no way to rule out the possibility.}
Theorem~\ref{thm:languagesofmind} makes this idea precise.

\noindent \begin{tabular}{p{0.30\textwidth}p{0.16\textwidth}}
\langE & \langT
\end{tabular}

\begin{theorem}\label{thm:languagesofmind}

For a given $\N$ and $\E$, let $\mathfrak{R}_{\N,\E} = \{\R : (\R, \N)\in \E\}$ be the set of respondents that are consistent with $\N$ and $\E$.
For $\E' \succeq \E_\lang$, if $\exists \N$ such that $\mathfrak{R}_{\N,\E'}\neq \emptyset$ and
$$ 
\bigcap_{\R\in \mathfrak{R}_{\N,\E'}} L_{R.z} = \emptyset,
$$
then there does not exist a demonstrable $\V$ which entails $\T_\lang$ with respect to $\E'$. 
\end{theorem}

Note that the hypothesis of Theorem \ref{thm:languagesofmind} implies that the government has no ability to check membership in the language $L_{\R.z}$.  This is the formalization of the property stated above, that the government does not know ``what it is looking for.''

We defer the proof of Thm.~\ref{thm:languagesofmind} to Appendix~\ref{app:proof:languagesofmind}.
The proof does not rely on any computational limitations on the part of the government---only the information-theoretic uncertainty of the true value of $\R.z$.

\triplealgorithm{\hybridD \hybridT}{\hybridE}{\hybridDprime}{\label{algs:hybrid}Algorithms for the example in Section~\ref{sec:entailment:partial-spec} to show the challenge of entailment with partial evidence}

\iffull
\subsection{Impossibility of entailing distributions}

In Appendix \ref{ex:commitment} we will investigate actions that result in a distribution, such as compelled encryption or commitments.
In those sections, it will be useful to make use of the following result, which states that it is impossible to entail actions that result in a distribution using the random coins of the action itself.

\begin{theorem}\label{thm:random}
Let $\E$ be some evidence, and 
let $\T_{\rand}$ be some target action.
Suppose that $\T_{\rand}$ uses its own randomness in a non-trivial way: specifically, there exists $\E$-consistent $\N,\R$ with the property that at least one fixed setting of the random tapes of $\N,\R$ have the property that $\vert \textsf{Support}(\T_{\rand}^{\N,\R}) \vert \ge 2$.
Then, there is no demonstrable $\V$ that entails $\T_{\rand}$.
\end{theorem}

We defer the proof of Theorem \ref{thm:random} to Appendix~\ref{app:proof:random}.
\fi

\subsection{Entailment and partial specifications}
\label{sec:entailment:partial-spec}
The ability for evidence to partially-specify functionalities in $\N$ poses a major difficulty for proving entailment in many cases. 

The algorithms in Figure~\ref{algs:hybrid} illustrate the problem. Consider $\E_\hybrid$ stating that $D_\hybrid \prec \N[\loc{D}]$: there is a device at location $\loc{D}$ in nature that is consistent with $D_\hybrid$. $D_\hybrid$ has a variable $\msg$ and a method $\Read()$ that returns $\msg$.
It may seem that $\T_\hybrid$ which simply outputs $\N[\loc{D}].\Read()$ should be entailable. After all, $\P$ could just call $\N[\loc{D}].\Read()$ itself. Whatever $\V$ entails $\T_\hybrid$ could have an exemplar action that does nothing.

But in fact $\T_\hybrid$ is not entailable! The problem is that the behavior of $\N[\loc{D}]$ is only partially specified. While it must be consistent with $D_\hybrid$, it may have other methods too. For example, $\N[\loc{D}]$ could instead contain device $D_\readwrite$ (Figure~\ref{algs:hybrid}) which has another method $\Write(x)$ that sets $\msg$ to $x$. 
Because $\E$ contains no information about $\msg$, if the $\xA$ that does nothing is $\V$-conforming, then $\A_x$ that calls $\Write(x)$ is also $\V$-conforming.

As such, many of our entailment claims will be proved with respect to the stronger evidence $\E_D$ wherein some function specifications are strengthened from partial to full specifications (see Definition~\ref{def:E-D}).
This showcases our framework's ability to reason about uncertainty about the world, as illustrated by the example of deniable encryption in Section~\ref{sec:enter-password:deniable}.

\section{Compelled acts of computation}
\label{sec:examples}

In this section we consider how to construct demonstrable verifiers and analyze entailment for some important idealized foregone conclusion scenarios.
We begin with the common example of compelling a respondent to enter a passcode that decrypts a device under the assumption that the encryption scheme is not deniable (Section~\ref{sec:enter-password}).
Then in Section~\ref{sec:enter-password:deny} we show that removing the assumption of non-deniability eliminates the ability to entail the action of ``honest'' decryption, standing in contrast with most prior legal analysis.
We then describe the entailability of cryptographic \emph{decommitment}, which implies entailment of a wide variety of functionalities (Section~\ref{ex:decommitment}).
This finding stands in direct contrast with the approach of Scheffler and Varia \cite{USENIX:SchVar21} (see Appendix \ref{sub:compare-sv}).
We also examine additional examples, namely two-factor authentication, hash preimages, \emph{en}cryption, and commitmnets in Appendix~\ref{app:examples}.

\subsection{Entering a password to decrypt}
\label{sec:enter-password}\label{sec:entailment:decryption}

In the archetypal compelled decryption case, the government has lawfully seized an encrypted device and seeks order compelling the respondent to decrypt by entering her password.
The government wants to specify an action for which all implicit testimony is foregone. For this act to be meaningful, it should allow the government to recover the decrypted plaintext.
In our framework, the government seeks a demonstrable $\V_\enterPwd$ that entails the target task $\T_\enterPwd$ of decrypting and producing the plaintext (given in Figure~\ref{algs:enter-a-password}). 
For an idealized non-deniable encryption, we will show that this is possible if the government knows that the respondent knows the password (as in Kerr \cite{kerr2018compelled}).

In this subsection, we will consider the algorithms and evidence given in Figure~\ref{algs:enter-a-password} (unless otherwise stated).
In particular, we consider a simple password-protected device $D_\enterPwd$ that provides read access to a message if the correct password is entered, and hides all message contents otherwise.
If the government has evidence $\E_\enterPwd$ that the respondent knows this password, then it can use a verifier $\V_\enterPwd^{\N}$ that checks if the device is unlocked and capable of displaying any message, along with an exemplar action $\xA_\enterPwd^{\N,\R}$ in which the respondent types the known password into the prompt.

\subsubsection{$\R$ knows the password}
We begin with the case where $\G$ knows that $\R$ knows the password.
The evidence $\E_\enterPwd$ states that $\G$ knows that a device implementing $D_\enterPwd$ is at location $\loc{D}$ in nature. 
$D=D_\enterPwd$ is some idealized encrypted device, with a hardcoded password $D.\pwd$ and message $D.\msg\neq \bot$.
The evidence also states that $\R$ has some method $\R.\pwd()$ that returns the password $D.\pwd$ (the starred line in Fig.~\ref{algs:enter-a-password}). After that password is entered into $D$, the message $D.\msg$  can be read using the method $D.\Read()$.

Applying our framework, the government must specify a demonstrable verification procedure $\V_\enterPwd$. 
The government's verifier $\V_\enterPwd$ reads $\msg$ from the device, outputting $\Accept$ if {$\msg \ne \bot$}.

\begin{claim}
$\V_\enterPwd$ is demonstrable with respect to $\E_\enterPwd$.
\end{claim}

\begin{proof}
Consider exemplar action
$\xA_\enterPwd$ in Fig.~\ref{algs:enter-a-password}. First, we must check that every method call by $\xA_\enterPwd$ to $\R$ produces some output. $\xA_\enterPwd$ only calls $\R.\pwd()$. By $\E$, this method exists and outputs $D.\pwd$. 
Second, we must check that $\V_\enterPwd$ accepts $\xA_\enterPwd$ with respect to $\R$, $\N$. Let $M$ be the device at $\N[\loc{D}]$. \E guarantees that $M$ implements $D_\enterPwd$ (i.e., $D_\enterPwd \prec M$). $\xA$ calls $M.\prompt$ with input $\R.\pwd()$. By \E, $\R.\pwd() = M.pwd$. Hence, when $\xA$ halts, $M.\decrypted = \True$. Now $\V_\enterPwd$ calls $M.\Read()$, which returns $M.\msg$. \E states that $\msg \neq \bot$.  Hence $\V_\enterPwd$ returns $\Accept$.
\end{proof}

\triplealgorithm{\eptdE  \eptdV }{\eptdD}{\eptdX {\eptdT  \eptdP}}{\label{algs:enter-a-password}Example algorithms for compelled decryption by entering a password as described in Section~\ref{sec:enter-password} and used in Section~\ref{sec:enter-password:deny}. The starred ($\star$) assertion in $\E_\enterPwd$ refers to the respondent's knowledge of the decryption password. Removing that line captures the setting when the government has no evidence of the respondent's knowledge of the password.}

By construction, any $\V_\enterPwd$-conforming action $\A$ will allow the government to recover the plaintext $\msg$, assuming the encrypted device is not deniable. We use entailment to make this precise. Consider $\T_\enterPwd$ in Figure~\ref{algs:enter-a-password}, which enters $\R.\pwd()$ into the device's password prompt, then reads and returns the message $\msg$. By the evidence, the result is the plaintext message.

We prove that $\V_\enterPwd$ entails $\T_\enterPwd$ under the stronger evidence $\E_{D_\enterPwd}$ that the device is fully specified by $D$---ruling out deniability. 
This restriction is necessary, and follows the general pattern described in Section~\ref{sec:entailment:partial-spec}.
For one, the proof breaks down: the step marked ($\clubsuit$) does not hold for arbitrary $\E$.  Moreover, Section~\ref{sec:enter-password:deniable} shows that the government cannot entail $\T_\enterPwd$ when the device is deniable but otherwise consistent with $D_\enterPwd$.

This example illustrates the expressiveness of our approach and the power of entailment.
Claim~\ref{claim:entailment-enter-password} states that in this example `decrypting by entering a password' entails  `decrypting then producing the plaintext'.
But the act of decrypting-then-producing is not itself compellable! Specifically, it would require implicit testimony about conformity that is not a foregone conclusion under the evidence.
(Additionally outside of our model, some would argue that producing the plaintext presents testimonial issues beyond \emph{Fisher}.)

\begin{claim}
\label{claim:entailment-enter-password}
$\V_\enterPwd$ entails $\T_\enterPwd$ with respect to $\E_{D_\enterPwd}$.
\end{claim}

\begin{proof}
We must provide an efficient oracle machine $\P_\enterPwd$ such that for all $\A$ that is $\V_\enterPwd$-conforming: $\P_\enterPwd^{\N'}(\trans) = \T_\enterPwd^{\N,\R}$.
Let $M = \N[\loc{D}]$ and $M' = \N'[\loc{D}]$.
Let $\P_\enterPwd^{\N'}$ simply call and return $M'.\Read()$ as in Figure~\ref{algs:enter-a-password}. To complete the proof, we show that both $\P_\enterPwd$ and $\T_\enterPwd$ always output $M.\msg$.

\emph{(RHS)}.
By $\E_\enterPwd$-consistency of $\R$, we get $x = M.\pwd$ on line 1 of $\T_\enterPwd$. By $\E$-consistency of $M$, we get $\msg' =M.\msg$ on line 3 of $\T_\enterPwd$. Hence $\T_\enterPwd^{\N,\R}$ always outputs $M.\msg$.

\emph{(LHS)}.
Because $\A$ is $\V_\enterPwd$-conforming, we have $\exec{\V_\enterPwd^{\N}}{ \A^{\N,\R} }=\Accept$. 
By construction, the execution $\exec{\cdot}{\cdot}$ first runs $\A$ and then runs $\V_\enterPwd$.
Because $\V_\enterPwd$ does not change the state of $M$ (i.e., does not call $M.\prompt$), $\N'$ is equal to the state of nature after $\A^{\N,\R}$ terminates.
By construction of $\V_\enterPwd$ (and the conformity of $\A$), $M'.\Read() \neq \bot$ after $\A$ terminates. 
If $M'.\Read()\neq \bot$, then $M'.\Read()=M.\msg$  \ ($\clubsuit$). Hence $\P_\enterPwd^{\N'}$ always outputs $M.\msg$.
\end{proof}

\subsubsection{$\R$ may not know the password}

Next, we drop the starred assertion in $\E_\enterPwd$ in Figure~\ref{algs:enter-a-password}. This tweaks the facts to remove the government $\G$'s assertion that $\R$ knows a password to $D$. 
Call the modified evidence $\E^\star$.

$\E^\star$ states that the respondent has some method $\R.\pwd()$ that may do something. But not what it does, nor whether it has any relation to $D$.
Under $\E^\star$, it is entirely consistent that $\R$ simply does nothing.
That is, for any $(\N,\R)$ that is $\E^\star$-conforming, $(\N,\R_\bot)$ is also $\E^\star$-conforming where $\R_\bot$ immediately halts on all inputs.

\begin{claim}
$\T_\enterPwd$ is not entailable with respect to $\E^\star$.
\end{claim}

\begin{proof}
Let $(\N,\R_\enterPwd)$ be $\E_\enterPwd$-consistent. Because $\E_\enterPwd \succeq \E^\star$,  $(\N,\R_\enterPwd)$ is $\E^\star$-consistent. Hence $(\N,\R_\bot)$ is also $\E^\star$-consistent, where $\R_\bot$ immediately halts on all inputs. Fix $\V$ demonstrable and  $\A$ that is $\V$-conforming, all with respect to $\E^\star$. Because $\R_\bot$ is $\E^\star$-consistent, the action $\A_\bot$ which emulates $\A$ with $\R_\bot$ is also $\V$-conforming with respect to $\E^\star$.

For any post-processor $\P$, the output of $\P$ after $\exec{\V^\N}{\A_\bot^{\N,\R}}$ is independent of $\R$.
On the other hand, $\T_\enterPwd^{\N,\R_\enterPwd}$ always outputs $D.m\neq \bot$, while $\T_\enterPwd^{\N,\R_\bot}$ always outputs $\bot$.
Hence, there exists $\E^\star$-consistent $(\N,\R)$ such that $\P$ does not always output the result of $\T_\enterPwd^{\N,\R}$ (e.g., $\R$ is an appropriately chosen distribution over $\R_\bot$ and $\R_\enterPwd$).
Hence, $\T_\enterPwd$ is not entailable with respect to $\E^\star$.
\end{proof}

\subsection{Deniable encryption}
\label{sec:enter-password:deniable}
\label{sec:enter-password:deny}

We continue the example from Section~\ref{sec:enter-password} and Figure~\ref{algs:enter-a-password} wherein the government wishes to compel decryption by entering a password. Now we consider the case of deniable encryption.
Deniable encryption introduces a \emph{duress} password which, when entered, allows a device to be decrypted to something other than its ``true'' contents. For illustrative purposes, we analyze a simple but unrealistically powerful form of deniable encryption $D_\denysub$, wherein the duress password can be used to overwrite the device's contents arbitrarily.

\eptdDdeny

In Claim~\ref{claim:entailment-enter-password}, we proved entailment only with respect to the strong evidence $\E_{D_\enterPwd}$ that states that the encryption is \emph{not} deniable.
This is not a fluke. 
Consider the idealized version of deniable encryption in $D_\denysub$. Because $D_\enterPwd \prec D_\denysub$, it is entirely consistent with $\E_\enterPwd$. So the possibility that the actual encrypted device uses a form of deniable encryption affects entailment. Absent  stronger evidence that the device is not deniable (i.e.: $\E_{D_\enterPwd} \succeq \E_\enterPwd$), $\T_\enterPwd$ is not entailable, and the message that will ultimately be read from $D$ is not meaningful.

While entailment fails, demonstrability does not. The same $\V_\enterPwd$ is demonstrable with the same exemplar action $\xA_\enterPwd$. Using our framework, a respondent may ultimately be compelled to enter a password \emph{even if the government knows the encryption is deniable}---but the respondent would be free to use the duress password.
This is in contrast to prior legal approaches to this fact pattern, as discussed at the end of Section~\ref{sec:intro}.

\subsection{Opening a commitment}
\label{ex:decommitment}

Next, we consider \emph{commitments}---cryptographic objects related to, but distinct from encryption.
In this section, we ask: when can the government compel the opening of a commitment?
In our framework, the answer is: when the government has the commitment and knows that the respondent can open it.
Note that this answer generalizes the claims about committing encryption from Section~\ref{sec:enter-password}.

At a high level, commitments are like lockboxes, where each box has a bespoke lock and key. Once you put a secret message into the lockbox, nobody can read the secret (hiding), and nobody can change the secret (binding). Whoever has the key can open the box in front of others, thereby convincing them that box always contained that particular secret message.
An important difference between commitments and encryption is this public-opening property. For instance, deniable encryption schemes are specifically non-committing: it can be decrypted in different ways. Commitments can only be opened a single way.

Moving from the physical analogy to the computational setting, commitments schemes provide an algorithm to \emph{commit} to a secret $x$, producing the \emph{commitment} $c$ and the \emph{decommitment} $d$. \emph{Opening} the commitment means revealing $d$ and $x$. Given $c$, $d$, and $x$, anybody can \emph{check} that the opening was done correctly.

\begin{definition}[String commitment scheme]
\label{defn:commitment}
Let $(\Com, \Check)$ be a string commitment scheme where $(c,d) \gets \Com(x; r)$ yields a pair of strings $(c,d)$ called the \emph{commitment} and \emph{decommitment} to message $x$ respectively, and $\Check(c', d', x')$ returns either 1 or 0.  For all $r$, if $(c, d) \gets \Com(x; r)$ then $\Check(c, d, x) = 1$.

\textbf{Binding (informal):}
A commitment scheme is \emph{binding} if it is difficult to find $(x, x', d, d', c)$ such that $x \ne x$, $\Check(c, d, x)=1$, and $\Check(c, d', x') = 1$.
A commitment scheme is \emph{perfectly binding} if there does not exist such a tuple $(x, x', d, d', c)$.

\textbf{Hiding (informal):}
Suppose $(c, d) \gets \Com(x; r)$.
A commitment scheme is \emph{hiding} if it is difficult to find $x$ given only $c$.
A commitment scheme is \emph{perfectly hiding} if the distribution of $c$ is identical for any $x, x'$ input into $\Com$ (over the choice of $r$).
\end{definition}

\triplealgorithm{\decomevidence}{\decomverifier \decompost}{\decomtarget \decomexemplar}{ \label{fig:decommitment}Evidence, Verifier, Target Action, and Exemplar Action for decommitment (Section \ref{ex:decommitment}).
We define $\E_\binding \succeq \E_\exDecom$ to be the stronger evidence where the commitment scheme on the starred ($\star$) line is \emph{perfectly binding}.}

Consider the Evidence $\E_{\exDecom}$, Verifier $\V_{\exDecom}$, Target Action $\T_{\exDecom}$, and Exemplar Action $\xA_{\exDecom}$ given in Figure \ref{fig:decommitment}.
$\E_\exDecom$ states that that the respondent is able to produce a valid opening to the commitment $c$ at location $\N[\loc{comm}]$. The verifier $\V_\exDecom$ receives $(x', d')$ from the respondent (by way of a conforming action $\A$), fetches $c$ from nature, and runs the commitment's $\Check$ algorithm. The exemplar $\xA$ outputs $(x,d)$ produced by $\R$.

\begin{claim}$\V_{\exDecom}$ is demonstrable with respect to $\E_{\exDecom}$ with exemplar action $\xA_{\exDecom}$.\end{claim}
\begin{proof}
If $\N,\R$ are $\E_{\exDecom}$-consistent, then $\xA_{\exDecom}$ will result in $\V$ accepting, because the Evidence states that $\Check(C.\Read(),$ $\R.decom(),$ $\R.secret()) = 1$.  We know $\R$ has the ability to perform this action, because it involves only a read to $\R.secret()$ and $\R.decom()$ which were specified in $\E_{\exDecom}$.  Thus, $\V_{\exDecom}$ is demonstrable with witness action $\xA_{\exDecom}$.  
\end{proof}

In fact, this holds even without binding. That is, dropping the starred ($\star$) line in $\E_\exDecom$ in Figure~\ref{fig:decommitment} does not affect demonstrability.

If the government additionally has evidence that the commitment is perfectly binding, then $\V_\exDecom$ entails $\T_\exDecom$ which discloses $\R.\secret()$.\footnote{If the commitment scheme is only computationally binding, then we expect entailment would still hold against a computationally-bounded respondent along the lines of Claim~\ref{claim:hash}.}
Let $\E_\binding \succeq \E_\exDecom$ represent this stronger evidence. (Both $\T_\exDecom$ and $\E_\binding$ are defined in Fig.~\ref{fig:decommitment}.

\begin{claim}\label{claim:decom:entailment}
If $\Com$ is perfectly binding, then
$\V_{\exDecom}$ entails $\T_{\exDecom}$ under $\E_\binding$.
\end{claim}

\begin{proof}
Recall that $\E_{\exDecom}$ reads $c$ from $\N[\privatedatastyle{commLoc}]$.  $\E_{\exDecom}$ asserts that for $d \leftarrow \R.\methodstyle{decom}()$ and $x \leftarrow \R.\methodstyle{secret}()$, it must be true that $\Check(c, d, x)=1$.  Recall that $\V_{\exDecom}$ reads $c' \leftarrow \N[\E_{\exDecom}.\privatedatastyle{commLoc}]$, and since this oracle is read-only, we have that $c' = c$.
Since the commitment scheme is perfectly binding, there is no $(x, x', d, d', c)$ such that $x \ne x'$, and both $\Check(c, d, x)=1$ and $\Check(c, d', x')=1$.  Thus, $x' = x$ and $d' = d$.

Letting $\N'$ be the state of Nature after running $\A$, the post-processor $\P_{\exDecom}^{\R,\N'}$ which returns the $x'$ sent by $\A$ will always return $\R.secret()$, thus entailing $\T_{\exDecom}$.
\end{proof}

A corollary of Claim~\ref{claim:decom:entailment} is that, under the same evidence, the government can entail arbitrary computation on the committed secret $\R.\secret()$. For any $f$, $\T_f$ is entailable with respect to $\E_\exDecom$: 
$$\T_f^{\N,\R} := \text{\textbf{\textrm{return}}}~ f(\R.\secret())$$

\begin{remark}
In this section, we showed that compelling the opening of a unknown but committed secret is demonstrable and entailable if the scheme is perfectly binding.
\iffull
    In Appendix~\ref{ex:random}, we show that there is no way to entail a \emph{commitment} of an unknown secret.
\else
    In the full version, we show that there is now way to entail a \emph{commitment} of an unknown secret~\cite{fullversion}.
\fi
This is reassuring.
If it were not true, then $\G$ could  entail commitment-then-open of arbitrary unknown  secrets.
It is worth noting that there are many demonstrable $\V$s that have the \emph{exemplar action} of committing to a secret value.  However, the fact that none of those $\V$s \emph{entail} the action of committing to a secret value means that those $\V$s would also be satisfied by other actions, possibly for example committing to the all-zeros string.
\end{remark}

\section{Comparison to prior work}
\label{ssec:usenix}

\subsection{Brief comparison to other legal scholarship}\quad
Here, we compare
our verification-centric view
to the body of legal scholarship on compelled decryption.
On the one hand, scholars like
Kiok~\cite{kiok} and Sacharoff~\cite{sacharoff} posit that the government can compel decryption only if they know the desired files with ``reasonable particularity,'' and 
Winkler~\cite{winkler} argues that the self-incrimination privilege provides an absolute defense against compelled decryption.
On the other hand, Cohen and Park \cite{cohen-park}, Kerr~\cite{kerr2016,kerr2018compelled,kerr2020}, McGregor~\cite{mcgregor}, and Terzian~\cite{terzian2013fifth,terzian-dec} deem compelled decryption to be a foregone conclusion in some contexts,
either to re-balance government power due to the widespread use of encryption or
because all the testimony implicit in the act of decryption (e.g., ownership, knowledge of password) is foregone.

Our framework largely agrees with the latter interpretation of the foregone conclusion, by considering \emph{all acts} involved in the performance of decryption rather than simply the files produced.
That said, our framework reaches a novel conclusion about the compellability of ``deniable encryption'' schemes, which include special duress password(s) that can reveal different files.
Here, Kerr and Cohen-Park reach opposite conclusions.
By separating demonstrability from entailment, our framework reaches a ``middle ground'' where decryption can be compelled, but the respondent can use any password unless the files are known with reasonable particularity.

\subsection{Comparison to Scheffler and Varia}
\label{sub:compare-sv}
The work that is most related to ours is the paper of Scheffler and Varia \cite{USENIX:SchVar21}, which also contributes a cryptography-inspired candidate definition of the foregone conclusion doctrine.
This work and \cite{USENIX:SchVar21} agree on landmark Supreme Court and circuit court decisions to date about the foregone conclusion doctrine \cite{fisher_v_us,doe_v_us_ii,us_v_hubbell,us_v_greenfield,in_re_proceedings,in_re_subpoena,us_v_ponds,in_re_subpoena_2nd}. As a general rule of thumb that we expect to hold for many realistic fact patterns, an action that is a foregone conclusion under \cite{USENIX:SchVar21} is likely entailable in our current framework. A notable exception is the inability for our current framework to entail random sampling from a probability distribution (Theorem~\ref{thm:random}).

That being said, the two works differ in several legal and computer science aspects. 

First, this work requires the government to specify how it will \emph{verify} the respondent's act, thus ensuring that the government is not ``relying on the truthtelling'' \cite[l. 411]{fisher_v_us}.  In contrast, \cite{USENIX:SchVar21} requires the government to specify how it \emph{could have performed} the respondent's act using the cryptographic idea of simulation but without access to the ``contents of [the respondent's] mind'' \cite[l. 421]{fisher_v_us}, thus ensuring that the respondent's production ``adds little or nothing to the sum total of the Government's information'' \cite[l. 411]{fisher_v_us}.

Second, the two works take different approaches to conformity.  Scheffler and Varia dealt with the implicit testimony from conformity using simulation.  In contrast, this work defines away the problem of conformity by allowing the respondent to perform \emph{any} conforming action.

Third, the modeling in this work allows modeling more scenarios than \cite{USENIX:SchVar21}. Scheffler and Varia focused on acts of \emph{production}, and involved the simulation of \emph{what was produced}.  Our work, on the other hand, focuses on acts of either production or performance, and can analyze actions taken that have impact on the world even if no item is produced to the court as a result.  To accomplish this technically, rather than modeling the world as a collection of passive strings such as a ciphertext stored on a smartphone's hard drive, we model dynamic, stateful, interactive processes such as the phone itself and all of the ways that one can interact with it.

The archetypal compelled decryption scenario asks: When can a respondent be compelled to enter a password into an encrypted device in the government's possession? Scheffler and Varia model a scenario where the respondent performs password-based decryption and produces the plaintext. But Cohen and Park argue that the ``decryption'' scenario is doctrinally distinct from the ``production of plaintext'' \cite{cohen-park}.  In contrast to \cite{USENIX:SchVar21}, our current model allows us to consider either scenario directly.

\section{Conclusion \& Discussion}

In summary, we adopt a view of implicit testimony that treats ability and conformity as equally important, and we present a new verification-centric approach to compelling acts while constructively satisfying the foregone conclusion doctrine. On the technical side, we formally define the concepts of demonstrability and entailment and explore the compellability of acts of cryptography within our framework.
Our goal in this project was CS-\emph{and}-law: rigorous computer science that substantively engages with and contributes to legal thought.
In this section, we take a step back and reflect on what seemed to work for us.

This paper is the culmination of a five-year research agenda that alternated between primarily legal and primarily technical contributions.
First, AC in collaboration with Sunoo Park produced a law-first analysis grounded in their knowledge of cryptography; they showed how the application of the foregone conclusion doctrine can depend
on the nuances of the underlying technology \cite{cohen-park}.
Next, SS and MV produced a technical-first analysis inspired by the law; they
gave a more restricted simulation-based formal framework and
analyzed how the compellability of different cryptosystems under it \cite{USENIX:SchVar21}.
All three authors then worked to reconcile the differing accounts of those two works, resulting in a new legal analysis that introduced the idea that ability and conformity are testimony implicit in acts of performance \cite{plsc}.
The present paper represents our effort to turn these ideas into a formal framework; more than just a formalism, that effort spawned our verification-centric approach and the concepts of demonstrability and entailability.

In this project, neither side was subservient to the other. We did not start with a legal view or a set of technical outcomes to which we molded our framework.

Instead, our meta-level approach in this project was an iterative, alternating application of legal and technical thinking.
At a high level, our goal was to fill in a gap in our understanding of the foregone conclusion doctrine: namely, it's requirement that the government may not ``rely on the truthtelling'' of the respondent.
First, we would state a non-mathematical version of the requirement based on our understanding of the case law and fact patterns considered so far. 
Second, we attempted to formalize that statement in a technically precise and sound way, building from the framework of \cite{USENIX:SchVar21}.
Third, we would test the result against many fact patterns, whether uncontroversial act-of-production cases, existing compelled decryption cases, the cryptographic primitives studied in this paper, or hypotheticals specifically designed to stress the framework.
Fourth, we would ask whether the formalism uncovered a new aspect of the legal question that could help refine our thinking.
After many many iterations, we converged on the ideas in this paper.

We do not know what legal questions are amenable to the sort of CS-and-law analysis we undertake in this paper. We highlight some properties that we believe make compelled decryption fruitful for a cross-disciplinary study, in the hopes that it may help others find other fertile grounds.
First, law alone offers no obvious answers. Scholars are in disagreement and state supreme courts are split. And as Cohen and Park argue, forgone conclusion analyses can depend on technical details, and reasoning by analogy can only go so far.
Second, there are clear connections to well-studied concepts in cryptography and computer science. Of course the very artifacts in question issue are technological and cryptographic. And at a conceptual level, cryptographic community has spent decades stress-testing notions of simulation and verification. The intuitive connection to ideas in \emph{Fisher}'s --- adding to the government's knowledge, or relying on the respondent's truthtelling --- are immediate to a cryptographer (though making the intuition precise is not easy).
Third, the results of this project is of value to both communities, with the potential to affect future court cases and to spur new research directions in security and cryptography.

\section*{Acknowledgments}

Aloni Cohen and Mayank Varia were supported by the National Science Foundation under Grant No.\ 1915763 and by the DARPA SIEVE program under Agreement No.\ HR00112020021.
Mayank Varia was additionally supported by National Science Foundation Grants No.\ 1718135, 1801564, and 1931714.
Sarah Scheffler was supported by a Google Ph.D.\ Fellowship and the Center for Information Technology Policy at Princeton University.
We are grateful for the feedback from participants at the 2021 Privacy Law Scholar's Conference and 2022 CTIC Law \& Computer Science Roundtable.

\bibliographystyle{ACM-Reference-Format}
\balance
\bibliography{bib/extra,bib/abbrev3,bib/crypto,bib/usenix}

\appendix

\section{More compelled acts of computation}
\label{app:examples}

\iffull
\subsection{Two-factor authentication}
\label{ssec:mfa}

Many services and devices lock via \emph{two-factor authentication} (2FA) in which two ``forms'' of authentication are required.  A typical setting is for a service to require entering a passcode and entering an ephemeral code that the service sent to a second device controlled by the same individual.  
In this section, we extend the example in Section~\ref{sec:entailment:decryption} of decrypting a device by entering a password to require a code sent to a secondary device -- the location of which may not be known to the government, although they must know the respondent can access it. 

\begin{claim} \label{claim:mfa:demonstrable}
$\V_{\mfa}$ is demonstrable with respect to $\E_{\mfa}$ and exemplar action $\xA_{\mfa}$.
\end{claim}

\begin{proof}
To show demonstrability, we must show that $\R$ is capable of performing $\xA_{\mfa}$ and that $\V_{\mfa}$ accepts $\xA_{\mfa}$ with respect to all $\E_{\mfa}$-conforming $\R$ and $\N$.
The first half is trivial since $\xA_{\mfa}$ only calls methods of $\R$ that were declared in $\E_{\mfa}$.

To show the second half, first note that there is a device at location $\N[\loc{device}]$ that implements $D$. Call it $M^D$. We observe that since $\E_{\mfa}$ declares that $\R.\methodstyle{pwd}() == D.\pwd$, entering $\R.\methodstyle{pwd}()$ as input to $M^D.\promptpw$ is the same as entering $D.\pwd$, and thus the call to $\promptpw$ will set $\thecode$ and set $\gotpwd$ to $\True$.

The evidence also states that there is a device at another location $\N[\R.\findsecondary()]$ implementing $S$. Call it $M^S$. By the last assertion in $\E_{\mfa}$ we have that the call to $M^S.\getcode()$ will yield a code $c$ equal to $D.\thecode$.
Thus, entering that $c$ into $M^D.\promptcode$ will set $\decrypted$ to $\True$ (recall that $\gotpwd$ was set to $\True$ earlier).

Then when $\V_{\mfa}$ calls $M^D.\Read()$, the device will return $m \ne \bot$, and thus $\V_{\mfa}$ accepts as desired.
\end{proof}

\begin{claim}\label{claim:mfa:entailment}
$\V_{\mfa}$ entails $\T_{\mfa}$ with respect to the stronger evidence $\E_{D,S} \succeq \E_\mfa$ (i.e., that the specifications of $D$ and $S$ are full specifications).
\end{claim}

\begin{proof}
Much like the proof given in Section \ref{sec:entailment:decryption}, we will show that both running $\P_{\mfa}$ on the transcript of $\exec{\V_{\mfa}^{\N}}{\A^{\R,\N}}$, and running the target action $\T_{\mfa}$, yield the same result of $D.m$.

For the left hand side, observe that if $\V_{\mfa}$ returns $\Accept$, then $m \leftarrow \N[\E.\loc{device}].\Read()$ was not $\bot$.  $\E_{D,S}$ states that $\N[\E.\loc{device}]$ exactly implements $D$ (and similarly for $S$).
Hence, the only way for $\Read()$ to return anything other than $\bot$ was for it to return $D.m$, which cannot be altered using any of the methods provided.  Thus, if $\V_{\mfa}$ returns $\Accept$, the post-processor $\P_{\mfa}$ must also return $D.m$.

Consider the right hand side.  Because $\E_{D,S}$ specifies that $\R.\methodstyle{pwd}() = D.\pwd$, this ensures that in the first line of $\T_{\mfa}$, $x = D.\pwd()$.  Moreover, the second line of $\T_{\mfa}$ calls exactly the given code of $D.\promptpw(D.\pwd)$, which also calls exactly $\setcode(D.\thecode)$ no matter what the randomness tape of $D$ is.  The third line of $\T_{\mfa}$ sets $c$ to exactly $D.\thecode$, and so the fourth line of $\T_{\mfa}$ calls $D.\promptcode(D.\thecode)$.  Thus, $\decrypted$ is set to $\True$ and when $\T_{\mfa}$ calls $\Read()$, it must return $D.\msg$, as desired.

Thus, $\V_{\mfa}$ entails $\T_{\mfa}$ as desired.
\end{proof}
\fi

\subsection{Preimage of hash}
\label{ssec:hash}

We imagine a scenario in which $\G$ wishes to compel $\R$ to provide a file and verify the file using its \emph{hash}.
This example is captured in Figure~\ref{algs:hash}, where we envision $h$ to be a fixed hash function like SHA-256. 
By construction of the evidence, $\V_{\hash}$ is demonstrable. 
Entailment is more complicated.

\triplealgorithm{\hashE}{\hashV \hashP}{\hashX \hashT}{\label{algs:hash} Verify a produced file using a hash (see Section \ref{ssec:hash})}

\begin{claim}\label{claim:hash}
For all $\E_{\hash}$-consistent $\N, \R$ and for all $\V_{\hash}$-conforming $\A$, either
\begin{enumerate}[(i)]
\item $\P^{\N'}(\trans) = \T^{\R,\N}$, or
\item $\P^{\N'}(\trans)$ and $T^{\R,\N} = \N[\R.\findfile()].\Read()$ are a hash collision for $h$.
\end{enumerate}
\end{claim}

\begin{proof}
Let $x_t = \N[\R.\findfile()].\Read()$ be the output of $\T^{\R,\N}$. The evidence ensures that $h(x_t) = \E.y$. 
Let $x_p$ be the output of $\P^{\N'}(\trans)$. By definition, $x_p$ is the message sent by $\rA$ to $\V$. Because $\rA$ is $\V$-conforming, $\V(x_p) = \Accept$ and hence $h(x_p) = \E.y$.
Hence $h(x_p) = h(x_t)$. 
Either $x_p = x_t$ or not, corresponding directly to the two cases in the claim.
\end{proof}

In words, this verifier \emph{almost} entails $\T_{\R,\N}$, except that the government has not ruled out the possibility that $\R$ knows a hash collision at $y$.
Note that we have only shown that \emph{this} verifier almost but not quite entails the action -- a different verifier might very well entail the target action.

Since, in the real world, we do not expect the respondent to be capable of producing a hash collision, we expect that $\R$'s only strategy will be to send the exact file $\G$ wanted to entail.  By compelling a preimage of a hash (which $\G$ knows with certainty $\R$ is capable of doing), this forces the respondent to either respond with the file or a hash collision, putting the ball in their court.  Although the government did not have evidence that pinned down every degree of freedom the respondent had, $\G$ did not need to make any additional assumptions to create the demonstrable $\V$ shown, nor did $\G$ need to implement a potentially-more-costly verifier to find a different way of compelling the file that directly entailed the action.
This demonstrates why the focus of our system is on \emph{demonstrability} rather than entailment.

\iffull
\triplealgorithm{\mfaD \mfaS}{\mfaE \mfaV}{\mfaX \mfaT \mfaP}{\label{algs:mfa} Multi-factor authentication (see Section \ref{ssec:mfa})}
\fi

\iffull
\subsection{Compelling encryption and commitments}
\label{ex:random}
\label{ex:commitment}

Throughout this section, we will use the following evidences:
$\E_{\secret}$ corresponds to the situation where $\G$ knows $\R$ has some secret, but has no knowledge of what it is and no way to verify it.  $\E_{\known}$ corresponds to a situation where $\G$ is able to learn the secret independently of $\R$ (by checking $\N[\loc{x}]$).

\secretE
\knownE

\subsubsection{Compelling encryption}

Let $\OTP(k', m)$ return the bitwise XOR of $k'$ and $m$ assuming they are the same length.
Table \ref{tab:otp} shows whether there is any way of entailing the target action $\T = \OTP(k', \R.x)$ depending on which key $k'$ is used and on the evidence. Of the scenarios considered, $\T$ is entailable only when the government specifies a fixed key $k'$ in $\T$ and is able to recover the plaintext $\R.x$ from nature.

\begin{table}
\begin{tabular}{rccc}
\toprule
& {$k' = \R.k$} & {$k'$ fixed in $\T$} & {$k$ sampled in $\T$} \\
\midrule
$\E_{\secret}$ & NE (Thm. \ref{thm:languagesofmind}) & NE (Thm. \ref{thm:languagesofmind}) &  NE (Thm. \ref{thm:random}) \\
$\E_{\known}$ &  NE (Thm. \ref{thm:languagesofmind})  & E (Rem. \ref{rem:knownstring})  & NE (Thm. \ref{thm:random})  \\
\bottomrule
\end{tabular}
\caption{Is $\T = \OTP(k, \R.x)$ entailable (E) or not entailable (NE), for the given setting of $\E$ and generation of $k$?}
\label{tab:otp}
\end{table}

\begin{remark}\label{rem:knownstring}
The action which deterministically chooses $k'$ and then returns $\OTP(k', \R.x)$ is entailable under $\E_{\known}$. The verifier $\V$ that deterministically sets the same $k$ and then sets $m \leftarrow \N[\loc{x}].\Read()$ is demonstrable and must always have exact equality with the output of this action, and so entails the action.
\end{remark}

\begin{remark}\label{rem:randomkey}
Note that if the one-time-pad was replaced with a randomized encryption scheme, by Theorem \ref{thm:random} compelling the action of honestly encrypting $m$ under the known $k$ (using randomness in $\T$) is \emph{not} entailable by Theorem \ref{thm:random}.  However, this action \emph{is} entailable if $\G$ additionally fixes the randomness used in the encryption scheme, because $\V$ may check the exact equality of the result.  This latter result stands in contrast with Theorem 4.5.1 from Scheffler-Varia \cite{USENIX:SchVar21}, which found that compelling encryption under a freshly sampled key \emph{was} compellable (under a different definition as described in Section \ref{ssec:usenix}).
\end{remark}

\subsubsection{Compelling commitments}

Together, the two claims in this section show that there is no way to entail a commitment of an unknown secret.
Let $(\Com, \Check)$ be a 
string commitment scheme as defined in Defn.~\ref{defn:commitment}.
Recall $\E_{\secret}$ above.
Suppose $\T_{\textsf{com},\$}$ is the action which first sets $(c, d) \leftarrow \Com_h(\R.x; \$)$ and then outputs $c$. That is, $\T_{\textsf{com}}$ outputs a fresh commitment to $\R.x$.
\begin{claim}
\label{claim:comm:perfhiding:fresh}
There is no demonstrable $\V$ which entails $\T_{\textsf{com},\$}$ with respect to $\E_{\secret}$.
\end{claim}

\begin{proof}
Corollary of Theorem~\ref{thm:random}.
\end{proof}

Suppose instead that for a fixed string $r$, we consider the target action $\T_{\textsf{com},r}$ that first sets $(c, d) \leftarrow \Com_h(\R.x; r)$, and then outputs $c$.  This effectively ``derandomizes'' the $\T_{\textsf{com},\$}$ above.
\begin{claim}
\label{claim:comm:perfhiding:fixed}
If $\forall (c',d') \in \mathsf{Image}(\Com)$, $\exists x'$, $\forall d''$: 
$$(c', d'') \ne \Com(x', r),$$
then there is no demonstrable $\V$ which entails $\T_{\textsf{com},r}$ with respect to $\E_{\secret}$.
\end{claim}
\begin{proof}Corollary of Theorem~\ref{thm:languagesofmind}.
\end{proof}
\fi

\iffull

    \section{Proofs of Entailment Theorems}
    \label{app:proofs}
    
    \subsection{Proof of Theorem \ref{thm:languagesofmind}}
    \label{app:proof:languagesofmind}

\else

    \section{Proof of Theorem \ref{thm:languagesofmind}}
    \label{app:proof:languagesofmind}

\fi

\begin{proof}
Suppose for contradiction $\exists\E'\succeq \E_\lang$, $\exists\N$ as in the hypothesis, and $\exists$ demonstrable $\V$ entailing $\T_\lang$.
Let $\xA$ be the exemplar action guaranteed by 
demonstrability, and let $\P$ be the post-processor guaranteed by 
entailment.

We will construct $\R^*$ and $\A_0$ that violate the definition of entailment. We will have to show four things. First, $(\N,\R^*)$ are $\E'$-consistent. Second, that $\A_0$ is $\V$-conforming with respect to $(\N,\R^*)$. Third, that there is some set $L^*$ such that  $\T^{\N,\R^*}\in L^*$. Fourth, that with with non-zero probability $\P^{\N'}(\tau) = x^* \notin L^*$. 
This contradicts our hypothesis, proving the theorem.

Let $\R_0\in \mathfrak{R}_{\N,\E'}$.
Consider $\A_0$ the exemplar action corresponding to $\V$ with $\R_0$ hardcoded. That is, $\A_0$ emulates $\xA$ but replaces any messages to $\R$ with a message to an emulated $\R_0$.
By the definition of demonstrability, $\exec{\V^\N}{\A_0^{\N,\R}} = \exec{\V^\N}{\xA^{\N,\R_0}}$ returns $\Accept$ for all $\R$. Hence $\A_0$ is $\V$-conforming for all $\R \in \mathfrak{R}_{\N,\E'}$.

Next we show that the distribution of $\P^{\N'}(\tau)$ is independent of $\R$ when $\A = \A_0$.
For any $\R$, the execution $\exec{\V^\N}{\A_0^{\N,\R}} = \exec{\V^\N}{\xA^{\N,\R_0}}$ is completely independent of $\R$. This is because $\R$ cannot interact with $\N$ or $\V$ directly, and $\A_0$ doesn't communicate with $\R$ by construction.
Hence, the resulting state of nature $\N'$ and transcript $\tau$ are independent of $\R$, and hence $\P^{\N'}(\tau)$ is too.

Consider $x^* \in \mathsf{Support}(\P^{\N'}(\tau))$.
By the theorem's hypothesis, $\exists \R^* \in \mathfrak{R}_{\N,\E'}$ such that $x^*\notin L^*$, where $L^* \triangleq L_{\R^*.z}$. But $\E'$  implies that $\T^{\N,\R^*}\in L^*$.

$\R^*$ and $\A_0$ violate the definition of entailment as required, completing the proof by contradiction.
\end{proof}

\iffull
\subsection{Proof of Theorem \ref{thm:random}}
\label{app:proof:random}

\begin{proof}
Suppose by way of contradiction that there exists demonstrable $\V$ that entails $\T_{\rand}$.
Let $\xA$ be the exemplar action implied by demonstrability.
Let $\A_0$ be $\xA$ with the all-zeros randomness tape hardcoded. Because $\V$ is demonstrable,  $\V$ must accept even when $\xA$ is called with the all-zeros randomness tape, thus we know $\A_0$ is also $\E$-conforming.

Let $\N, \R$ have the property from the theorem statement, that is, there exists at least one setting of their randomness tapes for which $\vert \textsf{Support}(\T_{\rand}^{\N,\R}) \vert \ge 2$.  Fix their randomness tapes to the first such setting, we use $\N_0$ and $\R_0$ as shorthand for calling them with this specific randomness.

Now consider the entailment equation when considered with these fixed random tapes, that is, we require that for all settings of $\P$ and $\T_{\rand}$'s randomness,
$\P^{\N_0'}(\exec{\V^{\N_0}}{\A_0^{\R_0,\N_0}}) \equiv \T_{\rand}^{\R_0, \N_0}()$.
(Note that there is no issue arising from the fact that the left hand side uses state $\N_0'$ after running $\exec{\V^{\N_0}}{\A_0^{\R_0,\N_0}}$ rather than the original state $\N$; this argument will rely only on the fact that $\N_0$ is deterministic, not any other property of $\N_0$.)
By assumption, the right hand side has support size at least 2.

If $\P$ is deterministic, the left hand side has support size 1.  So for some setting of $\T_{\rand}$'s randomness, the two sides are not equal, which is a contradiction.
If $\P$ is randomized, then fix $\P_0$ as $\P$ called with the all-zeros random string.  Because our definition of entailment requires that the two sides be equal for all possible settings of $\P$ and $\T_{\rand}$'s randomness tape, the two sides must still be equal even for $\P_0$.  However, once again, the left hand side now has support size 1, and the right hand side has support size at least 2.  So there is at least one setting of $\T_{\rand}$'s randomness for which the two sides are not equal, which is a contradiction.

Thus, there is no demonstrable $\V$ that entails $\T_{\rand}$.
\end{proof}
\fi

\end{document}